\documentclass[11pt]{article}
\usepackage{graphicx}
\usepackage{latexsym}
\usepackage{amssymb}
\usepackage{amsmath}
\usepackage{amsthm}
\usepackage{thmtools}
\usepackage{mathtools}
\usepackage{forloop}
\usepackage[paper=letterpaper,margin=1in]{geometry}
\usepackage[ruled,vlined,linesnumbered]{algorithm2e}
\usepackage{paralist}
\usepackage{xcolor}\definecolor{lightgray}{gray}{0.9} 
\usepackage{todonotes}
\usepackage{tikz}
\usepackage[font=small]{caption}
\usepackage{hyperref}

\definecolor{darkgreen}{rgb}{0,0.5,0}
\hypersetup{
    unicode=false,            
    colorlinks=true,          
    linkcolor=blue!70!black,  
    citecolor=darkgreen,      
    filecolor=magenta,        
    urlcolor=blue!70!black    
}

\newtheorem{theorem}{Theorem}[section]
\newtheorem{lemma}[theorem]{Lemma}

\newtheorem{corollary}[theorem]{Corollary}
\newtheorem{definition}{Definition}[section]
\newtheorem{proposition}[theorem]{Proposition}
\newtheorem{observation}[theorem]{Observation}

\newcommand{\defcal}[1]{\expandafter\newcommand\csname c#1\endcsname{{\mathcal{#1}}}}
\newcommand{\defbb}[1]{\expandafter\newcommand\csname b#1\endcsname{{\mathbb{#1}}}}
\newcommand{\defvec}[1]{\expandafter\newcommand\csname v#1\endcsname{{\mathbf{#1}}}}
\newcounter{calBbCounter}
\forLoop{1}{26}{calBbCounter}{
    \edef\letter{\alph{calBbCounter}}
		\edef\Letter{\Alph{calBbCounter}}
    \expandafter\defcal\Letter
		\expandafter\defbb\Letter
		\expandafter\defvec\letter
}

\newcommand{\eps}{\varepsilon}
\newcommand{\nnR}{{\bR_{\geq 0}}}
\DeclareMathOperator{\trunc}{trunc}
\DeclareMathOperator{\rank}{rank}
\DeclareMathOperator*{\expect}{\mathbb{E}}
\newcommand{\mo}{m^{(o)}}

\newcommand{\barM}{\overline{\cM}}
\newcommand{\barv}{\overline{v}}
\newcommand{\barI}{\overline{\cI}}
\newcommand{\email}[1]{{\href{mailto:#1}{#1}}}
\newcommand{\ow}{{\bar{w}}}

\DeclareMathOperator{\Poly}{Poly}

\SetKwComment{Comment}{$\triangleright$\ }{}
\SetCommentSty{itshape}
\DontPrintSemicolon

\begin{document}

\title{Submodular Maximization over a Matroid $k$-Intersection: Multiplicative Improvement over Greedy}
\author{Moran Feldman\thanks{Department of Computer Science, University of Haifa. This work was done while the author was visiting Queen Mary University of London. E-mail: \email{moranfe@cs.haifa.ac.il}} \and Justin Ward\thanks{School of Mathematics, Queen Mary University of London. E-mail: \email{justin.ward@qmul.ac.uk}}}

\maketitle
\thispagestyle{empty}
\pagenumbering{Alph}
\begin{abstract}
We study the problem of maximizing a non-negative monotone submodular objective $f$ subject to the intersection of $k$ arbitrary matroid constraints. The natural greedy algorithm guarantees $(k + 1)$-approximation for this problem, and the state-of-the-art algorithm only improves this approximation ratio to $k$. We give a $\frac{2k \ln 2}{1 + \ln 2} + O(\sqrt{k}) < 0.819 k + O(\sqrt{k})$ approximation algorithm for this problem. Our result is the first multiplicative improvement over the approximation ratio of the greedy algorithm for general $k$. We further show that our algorithm can be used to obtain roughly the same approximation ratio also for the more general problem in which the objective is not guaranteed to be monotone (the sublinear term in the approximation ratio becomes $O(k^{2/3})$ rather than $O(\sqrt{k})$ in this case). 

All of our results hold also when the $k$-matroid intersection constraint is replaced with a more general matroid $k$-parity constraint. Furthermore, unlike the case in many of the previous works, our algorithms run in time that is independent of $k$ and polynomial in the size of the ground set. Our algorithms are based on a hybrid greedy local search approach recently introduced by Singer and Thiery~\cite{singer2025better} for the weighted matroid $k$-intersection problem, which is a special case of the problem we consider. Leveraging their approach in the submodular setting requires several non-trivial insights and algorithmic modifications since the marginals of a submodular function $f$, which correspond to the weights in the weighted case, are not independent of the algorithm's internal randomness. In the special weighted case studied by~\cite{singer2025better}, our algorithms reduce to a variant of the algorithm of~\cite{singer2025better} with an improved approximation ratio of $(k + 1) \ln 2 + O(\eps) < 0.694k + 0.694 + O(\eps)$, compared to an approximation ratio of $\frac{k+1}{2\ln 2} \approx 0.722k + 0.722$ guaranteed by Singer and Thiery~\cite{singer2025better}.

\medskip

\noindent \textbf{keywords:} submodular function, matroid $k$-parity, matroid intersection, local search, greedy
\end{abstract}


\newpage
\pagenumbering{arabic}
\section{Introduction} \label{sec:introduction}

An important branch of combinatorial optimization studies the maximization of linear and submodular functions subject to natural families of combinatorial constraints. One large such family is \emph{matroid $k$-intersection}---the family of constraints that can be represented as the intersection of $k$ matroid constraints. In $1978$, Fisher, Nemhauser and Wolsey~\cite{fisher1978analysis} showed that the natural greedy algorithm guarantees $(k + 1)$-approximation for the problem of maximizing a monotone submodular function subject to a matroid $k$-intersection constraint, and for linear functions the above approximation ratio improves to $k$~\cite{korte1978greedy,jenkyns1976efficacy}.

Improving over the above guarantees of the greedy algorithm has been an open question for many years, and until recently success has mostly been limited to special cases such as \emph{$k$-dimensional matching} constraints. Here, we are given a hypergraph $H=(V,E)$ such that the set $V$ of vertices can be partitioned into $k$ disjoint subsets (or parts) $V_1, V_2 \dotsc, V_k$, and each hyperedge $e \in E$ contains at most $1$ vertex from each one of the subsets $V_i$. Such hypergraphs are known as $k$-partite hypergraphs. The $k$-dimensional matching constraint requires that we select a set of vertex disjoint hyperedges from $E$ (i.e., a matching in $H$). This constraint can be captured by an intersection of $k$ partition matroid constraints, each enforcing that no vertex of a particular set $V_i$ is contained in two hyperedges.

Following a long line of works (including~\cite{arkin1998local,chandra2001greedy,berman2000approximation,berman2003optimizing,neuwohner2021improved,neuwohner2024limits,thiery2023thesis}), Neuwohner~\cite{neuwohner2023passing} obtained an approximation ratio of $0.4986k + O(1)$ for weighted $k$-dimensional matching (i.e., the problem of maximizing a linear function subject to a $k$-dimensional matching constraint).\footnote{We note that Thiery~\cite{thiery2023improved} derived better approximation results for small to moderate values of $k$, but for large $k$ values he only obtains a weaker guarantee of $0.5k$.} In the more general case of monotone submodular objectives, Ward~\cite{ward2012approximation} obtained a slightly weaker approximation ratio of $k/2 + O(1)$. In the less general case of a cardinality objective (i.e., when the goal is to maximize the number of elements in the solution), the best-known approximation is $k/3 + 
O(1)$~\cite{sviridenko2013large,cygan2013improved,furer2014approximating}. On the hardness side, Lee, Svensson and Thiery~\cite{lee2025asymptotically} have shown that even in the unweighted case, no polynomial time algorithm can obtain better than $k/12$-approximation, improving on a previous lower bound of $\Omega(k/\log k)$ due to Hazan, Safra and Schwartz~\cite{hazan2006complexity}. 

All the above algorithmic results apply even for the more general case of \emph{$k$-set packing} constraints, which are similar to $k$-dimensional matching constraints, but allow the graph $H$ to be a general $k$-uniform hypergraph rather than a $k$-partite hypergraph. Unlike $k$-dimensional matching constraints, $k$-set packing constraints are not captured by matroid $k$-intersection, but they have a different crucial property, namely, that multiple individual exchanges between two sets can be carried out simultaneously as a single valid exchange. This property holds for partition matroids and, more generally, for \emph{strongly base orderable} matroids, but does not necessarily hold even for a single \emph{general} matroid constraint. 

For general matroid $k$-intersection constraints, Lee, Sviridenko and Vondr{\'{a}}k~\cite{lee10submodular} gave a local search algorithm with an approximation ratio of $k - 1$ for linear functions and $k$ for monotone submodular functions. In the special case of a cardinality objective, the same authors~\cite{lee2013matroid} gave a different local search algorithm with an approximation ratio of 
$k/2 + \eps$. We note that while the improvement in approximation for the cardinality objective has been by a multiplicative factor of $2$, the improvements for submodular and linear functions have, until recently, been only by an additive term. This gap stems from the fact that the analysis presented in~\cite{lee2013matroid} is based on global arguments that are appropriate for the unweighted case in which all elements are equally valuable, but obtaining a result for a weighted problem requires more local arguments.

Recently, Singer and Thiery~\cite{singer2025better} have overcome this obstacle to give a $\frac{k}{2\ln2} + O(1) \approx 0.722k + O(1)$-approximation for linear objective functions by combining greedy and local search techniques. Like the result of Lee, Sviridenko and Vondr{\'{a}}k~\cite{lee2013matroid} for unweighted matroid $k$-intersection, the algorithm of~\cite{singer2025better} works in fact even for the more general class of \emph{matroid $k$-parity constraints}. This class may be viewed as a common generalization of matroid $k$-intersection and $k$-set packing constraints. See Figure~\ref{fig:results} for a graphical depiction of the relationships between the different families of constraints we have discussed. The figure also states the state-of-the-art result for each combination of constraint family and class of objective functions.

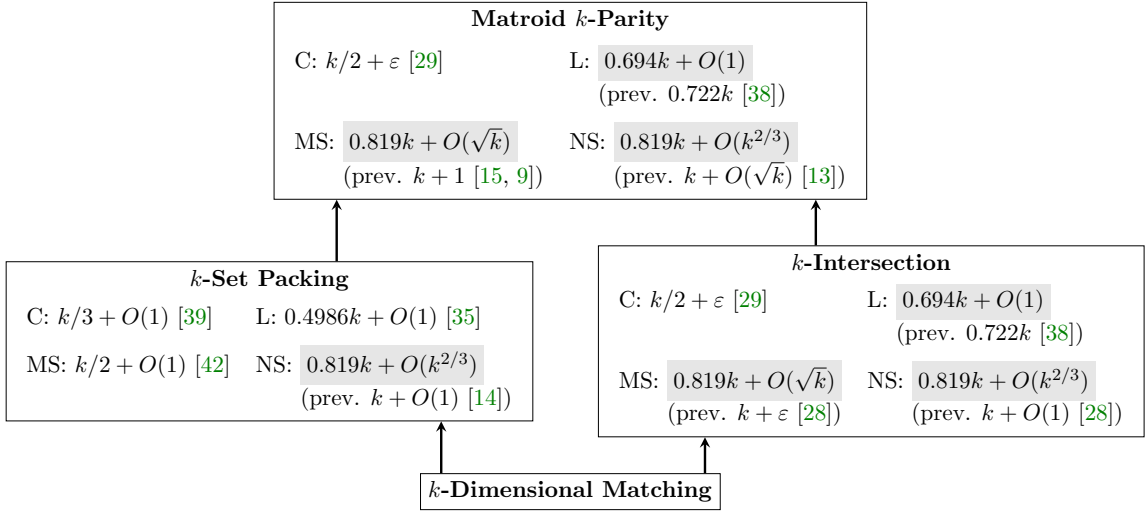
\begin{figure}
\begin{center}
\scalebox{0.8}{\begin{tikzpicture}
\node (k-DM) [draw, align=center] at (0, 0) {\textbf{$k$-Dimensional Matching}};

\node (k-SP) [draw, align=center] at (-5, 2.5) {\textbf{$k$-Set Packing} \\[2mm] \begin{tabular}{ll} C: $k/3 + O(1)$ \cite{sviridenko2013large} & L: $0.4986k + O(1)$ \cite{neuwohner2023passing} \\[2mm] MS: $k/2 + O(1)$ \cite{ward2012approximation} & NS: \colorbox{lightgray}{$0.819k + O(k^{2/3})$} \\ & \phantom{NS:} (prev. $k + O(1)$ \cite{feldman2011improved}) \end{tabular}};

\node (k-Intersect) [draw, align=center] at (5, 2.5) {\textbf{$k$-Intersection} \\[2mm] \begin{tabular}{ll} C: $k/2 + \eps$ \cite{lee2013matroid} & L: \colorbox{lightgray}{$0.694k + O(1)$} \\ & \phantom{L:} (prev. $0.722k$ \cite{singer2025better}) \\[2mm] MS: \colorbox{lightgray}{$0.819k + O(\sqrt{k})$} & NS: \colorbox{lightgray}{$0.819k + O(k^{2/3})$} \\ \phantom{MS:} (prev. $k + \eps$ \cite{lee10submodular}) & \phantom{NS:} (prev. $k + O(1)$ \cite{lee10submodular}) \end{tabular}};

\node (k-MP) [draw, align=center] at (0, 6.5) {\textbf{Matroid $k$-Parity} \\[2mm] \begin{tabular}{ll} C: $k/2 + \eps$ \cite{lee2013matroid} & L: \colorbox{lightgray}{$0.694k + O(1)$} \\ & \phantom{L:} (prev. $0.722k$ \cite{singer2025better}) \\[2mm] MS: \colorbox{lightgray}{$0.819k + O(\sqrt{k})$} & NS: \colorbox{lightgray}{$0.819k + O(k^{2/3})$} \\ \phantom{MS:} (prev. $k + 1$ \cite{fisher1978analysis,calinescu2011maximizing}) &  \phantom{NS:} (prev. $k + O(\sqrt{k})$ \cite{feldman2023how}) \end{tabular}};

\draw[-stealth,line width=0.4mm] (k-DM.north -| k-SP.335) -- (k-SP.335);
\draw[-stealth,line width=0.4mm] (k-DM.north -| k-Intersect.210) -- (k-Intersect.210);
\draw[stealth-,line width=0.4mm] (k-MP.south -| k-SP.50) -- (k-SP.50);
\draw[stealth-,line width=0.4mm] (k-MP.south -| k-Intersect.120) -- (k-Intersect.120);
\end{tikzpicture}}
\end{center}
\caption{The four families of constraints discussed in Section~\ref{sec:introduction}. An arrow from family A of constraints to family B indicates that family B generalizes family A. Below each constraint family, we list the state-of-the-art approximation ratios for maximizing cardinality (C), linear (L), monotone submodular (MS) and/or non-monotone submodular (NS) objectives subject to constraints of this family. New results of this paper are marked with a grey background (and the results they improve over appear below them in parentheses). To avoid repetition, the approximation ratios for $k$-Dimensional Matching constraints are omitted as they are identical to the ratios for the more general $k$-Set Packing constraints.} \label{fig:results}
\end{figure}

\subsection{Our Results} \label{ssc:our_results}
In this paper, we obtain the first result that multiplicatively improves over the approximation guarantee of the greedy algorithm for the problem of maximizing a monotone \emph{submodular} objective function subject to a \emph{general} matroid $k$-intersection constraint. Our technique is 
based on the one of Singer and Thiery~\cite{singer2025better}, and like theirs, works also for the more general family of matroid $k$-parity constraints. 

Let us briefly review the algorithm of Thiery and Singer.
Given some $\delta > 0$,\footnote{In our algorithm, we fix $\delta = 1/2$, but we leave it as a parameter here for consistency with~\cite{singer2025better}.} Thiery and Singer define a sequence of \emph{markers} (or thresholds), of the form $m_i = W(1-\delta)^{-i}$, where $W$ is the maximum weight of any element and $i \in \bZ_{\geq 0}$. These are used to partition the elements into classes according to their weight: an element $e$ appears in the $i$-th weight class if $w(e) \in [m_{i}, m_{i-1})$. The algorithm processes weight classes in descending order of weight and maintains a solution $A$, which is initially empty. When processing a weight class $C$, the algorithm applies local search to find a feasible approximately maximum weight extension $X \subseteq C$ to the current solution $A$. The elements of $X$ are then permanently added to $A$, and the algorithm continues to the next weight class. Intuitively, since all elements in $C$ have similar weights, one mainly cares about maximizing the number of elements in $X$. Let $X^*$ denote the largest possible extension $X$. Guarantees for the local search procedure for the unweighted matroid $k$-parity problem can be leveraged to show that $|X| \gtrsim (2/k) \cdot |X^*|$. However, the size of $X^*$ might be smaller than the number of elements of the optimal solution within the class $C$ because some elements of the optimal solution might be ``blocked'' by elements from heavier weight classes that are already in $A$. In order to bound the weight of the blocked elements, Thiery and Singer note that each selected element $e$ can block at most $k$ elements in total. Thus, if each blocked element has weight of at most $c \cdot w(e)$ for some $c < 1$, then one obtains $c\cdot k$ approximation for blocked elements. Unfortunately, $c$ might be very nearly 1, in which case the guarantee degenerates to $k$. However, this can happen only when the blocking element $e$ falls near the boundary of its weight class. To avoid this, Singer and Thiery multiply all markers by a single uniformly random shift $\tau$. They then show that for each element, the probability some marker $m_i$ lands near its weight is small, and thus the expected approximation for all blocked elements can indeed be bounded by $c \cdot k$ for some $c < 1$. 

A major obstacle for adapting the algorithm of Singer and Thiery to submodular functions is that, under such functions, there is no longer a static notion of weight. Previous works~\cite{lee10submodular,ward2012approximation,chakrabarti2015submodular,chekuri2015streaming} have adapted local search algorithms for weighted maximization to submodular objectives by using the marginal or incremental contributions of an element $e$ as a surrogate weight. In other words, if $f$ is the objective function, then the adaptations implicitly define the weight of $e$ to be $w(e) \triangleq f(A \cup \{e\}) - f(A)$, where $A$ is the current solution.\footnote{In some cases, $A$ is in fact some appropriate subset of the current solution, but this distinction is unimportant for the sake of this high-level discussion.} However, applying this idea to the algorithm of Singer and Thiery~\cite{singer2025better} breaks their solution for the problem of bounding the weight of the blocked elements. Recall that Singer and Thiery obtain this bound by applying a random shift, which guarantees that the probability that the weight of an element $e$ is close to the boundary of its weight class is small. However surrogate weights $w(e)$ defined as above depend on the solution $A$ constructed by the algorithm, and indirectly also on the random shift. This creates correlations between the weights and the shift, which make it difficult to bound the probability that the weight of an element is close to the boundary of its weight class.

Our solution for the above hurdle is to introduce an additional auxiliary weight $u(e)$ for every element $e$ in the optimal solution $O$. Unlike the surrogate weights $w(e)$, the weights $u(e)$ depend only on $O$, and not on the solution set $A$, and thus are independent of the random shift of the algorithm. However, as these weights depend on the unknown set $O$, they are unavailable to the algorithm, which must continue to use the weights $w(e)$. This limitation of the algorithm is not important when $w(e)$ and $u(e)$ roughly agree, and thus, the output of the algorithm is good (in expectation) in this case. In order to obtain a result that holds in general, we further show that whenever the two weights disagree significantly, this discrepancy can be used to obtain an alternative bound on the value of the output of the algorithm.

Along the way, we observe that since the weight classes are defined using exponentially decreasing thresholds, it is natural to draw the random shift used by the algorithm to set these thresholds from an exponential distribution. Using this distribution for the shift instead of the uniform distribution employed by~\cite{singer2025better} leads to some improvement also for linear objective functions compared to the approximation ratio of roughly $0.722k + O(1)$ obtained by~\cite{singer2025better}. More formally, we prove the following theorem.%
\begin{restatable}{theorem}{thmMonotoneResult}\label{thm:monotone_result}
For every $\eps > 0$, there exists an algorithm that runs in $\Poly(|E|, \eps^{-1})$ time and guarantees an approximation ratio of $\frac{2k \ln 2}{1 + \ln 2} + O(\sqrt{k}) \leq 0.819k + O(\sqrt{k})$ for maximizing a non-negative monotone submodular function $f\colon 2^E \to \nnR$ under a matroid $k$-parity constraint $(E, \cI)$. Moreover, if $f$ is linear, the approximation ratio of the algorithm improves to $(k + 1) \ln 2 + O(\eps) \leq 0.694k + 0.694 + O(\eps)$.
\end{restatable}

We note that the algorithm of~\cite{singer2025better} considers exchanges of size proportional to $k$ in its local search, which leads to a time complexity that is polynomial only when $k$ is constant. In contrast, the algorithm of Theorem~\ref{thm:monotone_result} considers only exchanges of constant size, and therefore, runs in polynomial time regardless of the value of $k$.

We also consider the maximization of a non-monotone submodular function subject to a matroid $k$-parity constraint. Algorithms for maximizing monotone submodular functions usually still produce for a non-monotone submodular function $f$ a guarantee of the form $f(A) \geq \alpha \cdot f(A \cup O)$, where $A$ is the output set of the algorithm, $O$ is the optimal solution and $\alpha$ is positive number. Gupta, Roth, Schoenebeck and Talwar~\cite{gupta2010constrained} described a general framework for translating such a guarantee into a real approximation ratio, and their framework was improved by subsequent works~\cite{mirzasoleiman2016fast,feldman2023how}. Unfortunately, because of the use of our novel auxiliary weights, when our algorithm from Theorem~\ref{thm:monotone_result} is analyzed in the context of a non-monotone submodular function $f$, it yields a guarantee of the more involved form $f(A) \geq \alpha \cdot f(A \cup O) - \beta \cdot f(O)$ for some $\beta < \alpha$. Nevertheless, we are able to adapt the version of the framework used by~\cite{feldman2023how} to transform guarantees of our more involved form into standard approximation ratios, which yields the following theorem.
\begin{restatable}{theorem}{thmNonMonotoneResult}\label{thm:non-monotone_result}
There exists an algorithm that runs in $\Poly(|E|)$ time and guarantees an approximation ratio of $\frac{2k \ln 2}{1 + \ln 2} + O(k^{2/3}) \leq 0.819k + O(k^{2/3})$ for maximizing a non-negative (not necessarily monotone) submodular function $f\colon 2^E \to \nnR$ under a matroid $k$-parity constraint $(E, \cI)$.
\end{restatable}

\subsection{Additional Related Work} \label{ssc:related_work}
When $k=2$, one can exactly optimize a linear function over a matroid $k$-intersection constraint via the classical algorithms of Edmonds~\cite{edmonds2003submodular} or Cunningham~\cite{cunningham1986improved}. Similarly, a $2$-set packing constraint is simply a standard matching constraint, and one can exactly optimize linear functions over it in polynomial time.  Lawler~\cite{lawler1976combinatorial} introduced the family of matroid parity constraints as a common generalization of both these families of constraints. Here, we are given a set $E$ of disjoint pairs of elements from $V$ and a matroid $\cM$ on $V$, and the goal is to find a collection of pairs from $E$ whose union is independent for $\cM$. In contrast to $2$-set packing and matroid $2$-intersection, it is impossible to maximize even a cardinality objective under a matroid parity constraint in polynomial time in the setting in which the matroid is given as an independence  oracle~\cite{lovasz1981matroid,jensen1982complexity}. Maximizing a cardinality objective under a matroid parity constraint is also known to be NP-hard (see, e.g.,~\cite[\S 43.9]{schrijver2003combinatorial}) and Lee, Sviridenko and Vondr\'{a}k~\cite{lee2013matroid} showed that it admits a PTAS.

When $k \geq 3$, optimizing over both matroid $k$-intersection and $k$-set packing constraints is NP-hard as both these constraint families generalize $k$-dimensional matching constraints, and optimizing over the last family of constraints is NP-hard even for cardinality objectives~\cite{karp72reducibility}. All these types of constraints are generalized by the family of matroid $k$-parity constraints, and this family also generalizes the above-mentioned family of matroid parity constraints in the sense that matroid parity constraints are equivalent to matroid $2$-parity constraints. The formal definition of matroid $k$-parity constraints can be found in Section~\ref{sec:preliminaries}, and we refer the reader to~\cite{lee2013matroid} for a more detailed discussion of the relationship between these and other families of constraints.

Finally, we note that improved results are known in several special cases. For matroid $2$-intersection constraints, Huang and Sellier~\cite{huang2023matroid} give a $\frac{3}{2}$-approximation in the special case that the objective is a vertex cover function. For matroid $k$-parity constraints, there are several results for the special case that the matroid is linear and a representation is given. In this setting, when $k=2$ an exact polynomial time algorithm for cardinality objectives was given by Lov\'{a}sz~\cite{lovasz1980matroid}, and 
Iwata and Kobayashi~\cite{iwata2022weighted} later gave an exact polynomial time algorithm for maximizing linear objectives. For larger $k$ values, in the same setting, exact FPT algorithms are known for both cardinality and linear objectives, with exponential dependence on $k$ and the rank of the matroid or the number of blocks in the solution~\cite{barvinok1995new,marx2009parameterized}. In contrast, in this paper, we give approximation algorithms that require no assumptions on the structure of the underlying matroid and have a running time that is independent of $k$.

\paragraph{Paper Structure.} Section~\ref{sec:preliminaries} describes the definitions and notation that we use, as well as some known results. Our main algorithm is then presented in Section~\ref{sec:local-search-algor}, and Section~\ref{sec:analysis} shows that this algorithm has the approximation guarantees stated by Theorem~\ref{thm:monotone_result} for linear and monotone submodular functions. Finally, Section~\ref{sec:non-monotone} explains how to use the main algorithm to prove also our result for non-monotone submodular functions (Theorem~\ref{thm:non-monotone_result}).


\section{Preliminaries}
\label{sec:preliminaries}

In this section, we present preliminaries consisting of definitions, notations and known results used in this paper. The preliminaries are loosely grouped based on whether they are more related to the objective function or to the constraint of the problem we consider.

\subsection{Constraint Preliminaries}

\paragraph{Matroids.} A combinatorial constraint is given by a pair $(E, \cI)$, where $E$ is a ground set of elements and $\cI \subseteq 2^E$ is the collection of all subsets of $E$ that are feasible according to the constraint. The constraint $\cM = (E, \cI)$ is a \emph{matroid} if $\cI$ is non-empty and also obeys the following two properties.
\begin{compactitem}
	\item Down-closedness: for every two sets $S \subseteq T \subseteq E$, if $T \in \cI$ then also $S \in \cI$.
	\item Augmentation: for every two sets $S, T \in \cI$, if $|S| < |T|$, then there exists an element $e \in T \setminus S$ such that $S \cup \{e\} \in \cI$.
\end{compactitem}

Like in the last definition, we often need in this paper to add or remove a single element from a set. For this purpose, given a set $S$ and an element $e$, it is convenient to use $S + e$ and $S - e$ as shorthands denoting $S \cup \{e\} $ and $S \setminus \{e\}$, respectively. We also need to introduce some concepts from matroid theory. It is customary to refer to the sets in the collection $\cI$ as the \emph{independent sets} of $\cM$. Notice that the augmentation property of matroids implies that all the inclusion-wise maximal independent sets of $\cM$ have the same size. This size is known as the \emph{rank} of $\cM$, and the (inclusion-wise maximal) independent sets of this size are the \emph{bases} of $\cM$. It is also useful to define for every set $S \subseteq E$ a value $\rank_\cM(S)$ that is equal to the largest size of an independent subset $S'$ of $S$. Notice that $\rank_\cM(E)$ is equal to the rank of the matroid $\cM$.

Given a matroid $\cM = (E, \cI)$, there are a few standard methods to produce a new matroid from it. In the following, we survey three of these methods.
\begin{itemize}
	\item Given a set $S \subseteq E$, the \emph{restriction} of $\cM$ to $S$ is a matroid over the ground set $S$ that is denoted by $\cM|_S$. The independent sets of $\cM|_S$ are all the independent sets $\cM$ that are subsets of this ground set. More formally, $\cM|_S = (S, 2^S \cap \cI)$.
	\item Consider a set $S \subseteq E$. \emph{Contracting} $S$ in the matroid $\cM$ produces a matroid over the ground set $E \setminus S$ that is denoted by $\cM / S$. Let $B$ be a maximal independent subset of $S$. Then, a set $T \subseteq E \setminus S$ is independent in $\cM / S$ if and only if $T \cup B \in \cI$.\footnote{There might be multiple maximal independent subsets of $S$, and therefore, multiple options for choosing $B$. It is known that $\cM / S$ is the same matroid for every such choice of $B$.}
	\item Given an integer $0 \leq r \leq \rank_\cM(E)$, the \emph{truncation} of $\cM$ to the rank $r$ is a matroid over the same ground set denoted by $\trunc(\cM, r)$. A set is independent in $\trunc(\cM, r)$ if it is independent in $\cM$ and its size is at most $r$. Notice that this implies, in particular, that the rank of $\trunc(\cM, r)$ is $r$.
\end{itemize}

We end the introduction of matroids by presenting the following known result about them, which is known as the Greene–Magnanti Theorem~\cite{greene1975some}.
\begin{lemma} \label{lem:partition}
Let $S$ and $T$ be bases of a matroid $\cM$. Then, for every partition $S_1, S_2, \dotsc, S_k$ of $S$, there exists a partition $T_1, T_2, \dotsc, T_k$ of $T$ such that $(S \setminus S_i) \cup T_i$ is a base of $\cM$ for every $i \in [k]$.
\end{lemma}

\paragraph{Matroid $k$-Parity.} A matroid $k$-parity constraint is given by a matroid $\cM = (V,\cI_{\cM})$ together with a partition $E$ of the elements of $V$ into disjoint sets of size at most $k$. Borrowing terminology from hypergraphs, we refer to the elements of $E$ as edges and the elements of $V$ as vertices. The ground set of the matroid $k$-parity constraint is the set $E$ of edges and a subset $S \subseteq E$ is feasible according to the constraint if the union of the edges in $S$ is independent in $\cM$. More formally, if we denote by $v(e)$ the set of up to $k$ vertices of $V$ appearing in $e$, then the matroid $k$-parity constraint is the combinatorial constraint $(E, \cI)$, where
\[
	\cI
	\triangleq
	\bigg\{S \subseteq E ~\bigg|~ \bigcup_{e \in S} v(e) \in \cI_\cM \bigg\}
	\enspace.
\]
It is often useful to extend the function $v$ to sets of edges in the natural way, namely, $v(S) \triangleq \cup_{e \in S}\; v(e)$ for every set $S \subseteq E$. Note that every matroid $k$-parity constraint is \emph{down-closed}; i.e., if $A \in \cI$ then $B \in \cI$ for all $B \subseteq A$. This follows immediately from the fact that the matroid $\cM$ is down-closed by definition.

In Section~\ref{sec:introduction}, it was mentioned that matroid $k$-parity constraints generalize the intersection of $k$ matroid constraints. Let us now briefly explain why that is the case. Consider an arbitrary combinatorial constraint $(X, \cI)$ that can be obtained as the intersection of $k$ matroids, i.e., there exist $k$ matroids $(X, \cI_1), (X, \cI_2), \dotsc,\allowbreak (X, \cI_k)$ such that $\cI = \cap_{i = 1}^k \cI_k$. We can represent this as a matroid $k$-parity constraint as follows. Define an edge for each element $x \in X$, and let the $k$ vertices of this edge be denoted by $(x,1), (x,2), \dotsc, (x,k)$. Intuitively, the element $(x, i)$ represents the element $x$ in the matroid $(X, \cI_i)$. Based on this intuition, we construct a matroid $\cM$ over the resulting set $V$ of vertices by defining that a subset $S \subseteq V$ is independent in $\cM$ if
\[
	\{e \in E \mid (e, i) \in S\} \in \cI_i \quad \forall\; i \in [k]
	\enspace.
\]
One can verify that $\cM$ is a matroid, and so $(E, \cI)$ is a matroid $k$-parity constraint given by $E$ and $\cM$ that is equivalent to the original matroid intersection constraint.

We use in this paper the exchange property for matroid $k$-parity constraints given by Lemma~\ref{lem:weights-main}. This property was implicitly proved by~\cite{singer2025better}, but we provide a different, and arguably simpler, proof of it. Lemma~\ref{lem:weights-main} proves that given two feasible sets $A$ and $B$, one can construct for every element $b \in B$ a set $N_b$ consisting of $A$ elements that intuitively block the addition of $b$ to $A$. Claims~\ref{item:emptyN} and~\ref{item:singletonN} of Lemma~\ref{lem:weights-main} show that for some sets $B' \subseteq B$, adding the elements of $B'$ to $A$ and removing the elements of $\cup_{b \in B'} N_b$ is an exchange that preserves feasibility. In general, designing sets $N_b$ that allow for such exchanges is not very difficult. For example, one could set $N_b = A \setminus B$ for all $b \in B$, which would allow an exchange for every $B' \subseteq B$. Thus, to make Lemma~\ref{lem:weights-main} interesting (and useful), it is necessary to have constraints on the structure of the sets $N_b$ guaranteeing that these sets are not ``too rich''. Claims~\ref{item:intersection} and~\ref{item:kN} of Lemma~\ref{lem:weights-main} play the role of these constraints.
\begin{lemma}
  \label{lem:weights-main}
Consider a matroid $k$-parity constraint $(E,\cI)$, and let $A, B \in \cI$ be two sets of edges that are feasible for this constraint. Then, there exists a collection of sets $\{N_b \subseteq A \mid b \in B\}$ such that
\begin{compactenum}
\item For every element $b \in B \cap A$, $N_b = \{b\}$; and for every element $b \in B \setminus A$, $N_b \subseteq A \setminus B$.\label{item:intersection}
\item $A \cup \{b \in B: N_b = \emptyset\} \in \cI$.\label{item:emptyN}
\item For every $a \in A$, $(A - a) \cup \{b \in B: N_b = \{a\}\} \in \cI$.\label{item:singletonN}
\item For every $a \in A$, there are at most $k$ distinct $b \in B$ such that $a \in N_b$.\label{item:kN}
\end{compactenum}
\end{lemma}
\begin{proof}
Let $\cM = (V, \cI_\cM)$ and $E \subseteq 2^V$ be the matroid and partition defining the given matroid $k$-parity constraint. Suppose first that $A \cap B \neq \emptyset$. Then, we can define a new matroid $k$-parity constraint based on the matroid $\cM / v(A \cap B)$ and the set of edges $E \setminus A \cap B$. Clearly, $A \setminus B$ and $B \setminus A$ are disjoint and feasible for this new matroid $k$-parity constraint. By applying the present lemma to these disjoint sets and this new constraint, we can obtain a collection of sets $\{N_b \subseteq A \setminus B \mid b \in B \setminus A\}$. Furthermore, adding to this collection the set $N_b = \{b\}$ for every $b \in A \cap B$ yields an extended collection of sets obeying all the properties required by the present lemma for the original sets $A$ and $B$ and the original matroid $k$-parity constraint. Thus, to prove the present lemma for general sets $A$ and $B$, it suffices to prove it for disjoint sets $A$ and $B$, which is our goal in the rest of this proof.

If $|v(A)| \leq |v(B)|$, then by repeated applications of the augmentation property of matroids, there is some set $C \subseteq v(B) \setminus v(A)$ with $|C| = |v(B)| - |v(A)|$ such that $v(A) \cup C \in \cI_\cM$. Similarly, if $|v(A)| \geq |v(B)|$, then there is some set $C \subseteq v(A) \setminus v(B)$ with $|C| = |v(A)| - |v(B)|$ such that $v(B) \cup C \in \cI_\cM$. Thus, in either case, there is a set $C$ such that $|v(A) \setminus C| = |v(B) \setminus C|$, $v(A) \cup C \in \cI_\cM$ and $v(B) \cup C \in \cI_\cM$. 

Define now the matroid $\barM \triangleq \trunc((\cM |_{v(A \cup B)}) / C, |v(A) \setminus C|) $, and let $\barI \subseteq 2^{V \setminus C}$ be the collection of independent sets of $\barM$. For each element $e \in A \cup B$, let $\barv(e) \triangleq v(e) \setminus C$, and similarly let $\barv(S) \triangleq v(S) \setminus C = \bigcup_{e \in S}\barv(e)$ for any $S \subseteq A \cup B$. Then, $\barv(A)$ and $\barv(B)$ are both bases of $\barM$---because they are independent in $\barM$ and their common size matches the rank of the truncation operation in the definition of $\barM$.  The sets $\{\barv(a) \mid a \in A\}$ form a partition of the base $\barv(A)$, and thus, by Lemma~\ref{lem:partition}, there exists a partition $\{\pi_a \subseteq \barv(B) \mid a \in A\}$ of the base $\barv(B)$ such that $|\pi_a| = |\barv(a)|$ and $(\barv(A) \setminus \barv(a)) \cup \pi_a \in \barI$ for every $a \in A$. This allows us to define $N_b \triangleq \{a \in A \mid \pi_a \cap \barv(b) \neq \emptyset\}$ for each $b \in B$. In the following, we prove that this collection of sets obeys all the claims stated by the lemma.

Claim~\ref{item:intersection} follows immediately from our assumption that $A$ and $B$ are disjoint. Let us prove next claim~\ref{item:emptyN}. Define $S = \{b \in B \mid N_b = \emptyset\}$. For every $b \in S$, it holds by definition that $\pi_a \cap \barv(b) = \emptyset$ for every $a \in A$. Since the sets $\pi_a$ are a partition of $\barv(B)$ and $\barv(b) \subseteq \barv(B)$, we must have $\barv(b) = \emptyset$ for every such $b \in S$. Thus,
\[
	v(A \cup S)
	\subseteq
	\barv(A \cup S) \cup C
	=
	\barv(A) \cup C
	=
	v(A) \cup C
	\in
	\cI_\cM
	\enspace.
\]

For the proof of claim~\ref{item:singletonN}, fix $a \in A$ and let $S = \{b \in B \mid N_b = \{a\}\}$. For all $b \in S$, $\barv(b) \subseteq \pi_a$ because otherwise $\barv(b)$ would have intersected additional sets in the partition of $\barv(B)$ beside $\pi_a$, which would have resulted in additional elements of $A$ appearing in $N_b$. Thus,
\begin{equation*}
\barv((A - a) \cup S) = (\barv(A) \setminus \barv(a)) \cup \barv(S) \subseteq (\barv(A) \setminus \barv(a)) \cup \pi_a \in \barI \enspace,
\end{equation*}
and therefore, $v((A - a) \cup S) \subseteq \barv((A - a) \cup S) \cup C \in \cI_\cM$.

Finally, for the proof of claim~\ref{item:kN}, fix $a \in A$ and recall that $v(b)$ and $v(b')$ are disjoint for every two distinct elements $b, b' \in B$. Therefore, $\pi_a \cap \barv(b) \subseteq v(b)$ and $\pi_a \cap \barv(b') \subseteq v(b')$ must also be disjoint for every two such elements. Thus,
\begin{equation*}
	| \{ b \in B \mid a \in N_b\}|
	=
	| \{ b \in B \mid \pi_a \cap \barv(b) \neq \emptyset \}|
	\leq
	\sum_{b \in B} |\pi_a \cap \barv(b)|
	=
	|\pi_a|
	=
	|\barv(a)|
	\leq
	|v(a)|
	\leq k
	\enspace.
	\qedhere
\end{equation*}
\end{proof}

\paragraph{Independence/Feasibility Oracle.} We conclude the discussion of combinatorial constraints by noticing that in most of their applications it is useful to have algorithms whose time complexities are polynomial in the size of the ground set $E$ rather than the size of the collection $\cI$ of feasible sets, which might be exponential in $|E|$. Therefore, algorithms that operate on such constraints usually assume that $\cI$ is not part of the input, but can be accessed through an oracle, known as \emph{independence} or \emph{feasibility oracle}, that takes as input a set $S \subseteq E$ and answers whether $S \in \cI$. The algorithms in this paper implicitly follow this standard assumption.

\subsection{Objective Function Preliminaries}

A set function $f\colon 2^E \to \bR$ is a function that assigns a value to every subset of a given ground set $E$. Given a set $S$ and an element $e$, we denote by $f(e \mid S) \triangleq f(S + e) - f(S)$ the marginal contribution of $e$ to $S$. Similarly, given an additional set $T$, we denote by $f(T \mid S) \triangleq f(S \cup T) - f(S)$ the marginal contribution of $T$ to $S$. The set function $f$ is called \emph{monotone} if $f(e \mid S) \geq 0$ for every set $S \subseteq E$ and element $e \in E$. The set function $f$ is called \emph{submodular} if $f(e \mid S) \geq f(e \mid T)$ for every two sets $S \subseteq T$ and element $e \in E \setminus T$. In other words, $f$ is submodular if the marginal contribution of an element $e$ to a set $S$ can only diminish when additional elements are added to the set $S$. An alternative, equivalent, characterization of submodularity states that $f$ is submodular if and only if $f(S) + f(T) \geq f(S \cup T) + f(S \cap T)$ for every two sets $S, T \subseteq E$.

In the case that $f(S) + f(T) = f(S \cup T) + f(S \cap T)$ for all $S,T \subseteq E$, we say that $f$ is \emph{modular}. Equivalently, a modular set function can be defined as a function of the form $f(S) = w_0 + \sum_{e \in S} w_e$, where $w_0$ and $\{w_e \mid e \in E\}$ are arbitrary real values. In the special case that $w_0 = 0$ (and thus, $f(\emptyset) = 0$), we say that $f$ is \emph{linear}.

Our algorithms optimize a submodular set function $f$. However, as in the case of combinatorial constraints, it often does not make sense to assume that the algorithm has an explicit description of $f$ because this description can be exponential in the size of the ground set $E$. Instead, our algorithms implicitly make the standard assumption that $f$ is accessed through a \emph{value oracle} that given a set $S \subseteq E$ returns $f(S)$.


\section{Main Algorithm}
\label{sec:local-search-algor}

In this section, we present our main algorithm, given below as Algorithm~\ref{alg:1}. In Section~\ref{sec:analysis}, we prove that Algorithm~\ref{alg:1} provides the approximation guarantees for linear and monotone submodular functions given in Theorem~\ref{thm:monotone_result}. Algorithm~\ref{alg:1} is also a central component in the proof of Theorem~\ref{thm:non-monotone_result} for non-monotone submodular objectives. As stated, Algorithm~\ref{alg:1} is not necessarily a polynomial time algorithm as it might have a large number of iterations. However, in Appendix~\ref{app:time_complexity} we describe a polynomial time implementation of this algorithm that skips over unnecessary iterations.

Algorithm~\ref{alg:1} gets as input a ground set $E$, a submodular function $f\colon 2^E \to \bR$ obeying $f(\emptyset) \geq 0$, the collection $\cI$ of the feasible sets of some matroid $k$-parity constraint over the ground set $E$ and an error parameter $\eps \in (0, 1)$. It then iteratively constructs sets $A_1, A_2, \dotsc$ whose union becomes the algorithm's output. For ease of notation, Algorithm~\ref{alg:1} uses $A_{\leq i}$ to denote the union $A_1 \cup A_2\cup \dotso \cup A_i$. When thinking in terms of the intuition given in Section~\ref{sec:introduction}, $A_1$ is the part of the solution of the algorithm chosen from the highest ``weight class'', $A_2$ is the part of the solution of the algorithm chosen from the second highest ``weight class'', and so on. Accordingly, Algorithm~\ref{alg:1} adds to the set $A_i$ only elements with large enough marginal values. More specifically, Algorithm~\ref{alg:1} defines a series of exponentially decreasing thresholds $m_i$, and adds to $A_i$ only elements whose marginal contribution is at least $m_i$.

To understand the details of the construction of the sets $A_i$ by Algorithm~\ref{alg:1}, we need the following definition.
\begin{definition}[$(\theta, \eps)$-improvement]
  \label{def:improvement}
Let $\theta > 0$ and $\eps \in (0,1)$ be two given values, and $A \in \cI$ be a feasible solution. We say that a pair of sets $S \subseteq E \setminus A$ and $N \subseteq A$ is a \emph{$(\theta, \eps)$-improvement for $A$} if $(A \cup S) \setminus N \in \cI$ and one of the following holds:
\begin{compactenum}
\item $S = \{x\}, N = \emptyset$, and $f(x \mid A) \geq \theta$.
\item $S = \{x\}, N = \{y\}$, $f(x \mid A) \geq \theta$, and $f(A + x - y) \geq f(A) + \eps \theta$.
\item $|S| = 2, N = \{y\}$ and $f(x_1 \mid A) \geq \theta$, $f(x_2 \mid A + x_1) \geq \theta$ for at least one of the two possible ways to assign the labels $x_1,x_2$ to the elements of $S$.
\end{compactenum}
\end{definition}
At the beginning of iteration $i$ of Algorithm~\ref{alg:1}, the sets $A_1, A_2, \dotsc, A_{i - 1}$ are already finalized. Algorithm~\ref{alg:1} then tries to construct a set $A_i$ that is a good addition to these sets. It does so by initializing $A_i$ to be the empty set, and then improving $A_{\leq i - 1} \cup A_i$ via a series of $(m_i,\eps)$-improvements. Notice that, by definition, the elements added by these improvements to $A_i$ have marginals of at least $m_i$, as promised above. Since the sets $A_1, A_2, \dotsc, A_{i - 1}$ are finalized before the $i$-th iteration of Algorithm~\ref{alg:1}, and must not be modified at this iteration, the algorithm considers only $(m_i, \eps)$-improvements whose set $N$ of removed elements is a subset of the current set $A_i$. Whenever an improvement of this kind is found, it is applied by setting $A_i \gets (A_i \cup S) \setminus N$, and once no more such $(m_i, \eps)$-improvements can be found, the $i$-th iteration of Algorithm~\ref{alg:1} terminates.


\begin{algorithm}
  Let $W \gets \max_{e \in E} f(e \mid \emptyset)$.\;
  Let $\alpha$ be a uniformly at random value from $(0,1]$, and let $\tau \gets 2^{\alpha}$.\;
  Define $m_i \triangleq W \tau 2^{-i}$ for every integer $i \geq 0$.\;
	\BlankLine

  Let $i \gets 0$ and $A_{\leq 0} \gets \emptyset$.\;
  \While{there exists $e \in E \setminus A_{\leq i}$ such that $f(e \mid A_{\leq i}) > 0$ and $A_{\leq i} + e \in \cI$\label{line:main_loop}}
  {
		Update $i \gets i + 1$.\;
		Let $A_{i} \gets \emptyset$.\;
    \While{\label{li:while}there is some $(m_{i}, \eps)$-improvement $(S,N)$ for $A_{\leq i - 1} \cup A_i$ with $N \subseteq A_i$}
    {
      $A_i \gets (A_i \setminus N) \cup S$.\label{li:whileend}\;
    }
    Let $A_{\leq i} \gets A_{\leq i - 1} \cup A_i$.\;
  }
  \Return{$A_{\leq i}$}.
\caption{\textsc{Greedy/Local-Search Hybrid Algorithm }$(E, f, \cI, \eps)$}
\label{alg:1}
\end{algorithm}

In Appendix~\ref{app:time_complexity}, we show that Algorithm~\ref{alg:1} can in fact be simulated in polynomial time. However, in order to derive the approximation guarantees of Theorem~\ref{thm:monotone_result}, it is sufficient to prove here the weaker claim that Algorithm~\ref{alg:1} terminates. In each iteration $i$, Algorithm~\ref{alg:1} repeatedly searches for and applies $(m_i, \eps)$-improvements for some current solution $A$. Each such improvement either increases $|A|$, or leaves $|A|$ unchanged and increases $f(A)$. Since there are only finitely many possible distinct values for $|A|$ and $f(A)$, it follows that each iteration of  Algorithm~\ref{alg:1} must terminate. The following observation implies that Algorithm~\ref{alg:1} can perform at most finitely many iterations before terminating. 
\begin{observation}
\label{obs:final-set-L}
For any input to Algorithm~\ref{alg:1}, there is an integer $L \geq 0$ such that $f(e \mid A_{\leq L}) \leq 0$ for every $e \in E \setminus A_{\leq L}$ with $A_{\leq L} + e \in \cI$, and furthermore, if $L \neq 0$, then $A_L \neq \emptyset$.
\end{observation}
\begin{proof}
First, we note that any iteration $i$ in which at least one $(m_i, \eps)$-improvement is found, must end with $|A_i| > 0$. This follows directly from the structure of $(\theta, \eps)$-improvements, which never decrease the size of the set $A$ that is being improved. Since the sets $A_i$ produced across different iterations are  disjoint by construction, there can be at most $|E|$ iterations in which the algorithm finds at least one $(\theta, \eps)$-improvement. Let $L$ be the index of the last such iteration, or $0$ if there are no such iterations. Suppose, for the sake of contradiction, that there is some $e \in E \setminus A_{\leq L}$ with $A_{\leq L} + e \in \cI$ and $f(e \mid A_{\leq L}) > 0$. Since no $(\theta, \eps)$-improvements are found after iteration $L$, we must have $A_{\leq i} = A_{\leq L}$ for all $i > L$. But, for sufficiently large $i \geq 1 + \log_2 W - \log_2 f(e \mid A_{\leq L})$, we then have $e \in E \setminus A_{\leq i}$ with $f(e \mid A_{\leq i}) = f(e \mid A_{\leq L}) \geq m_i$ and $A_{\leq i} + e = A_{\leq L} + e \in \cI$. Thus, $(\{e\}, \emptyset)$ would be a valid $(m_{i}, \eps)$-improvement in iteration $i > L$, contradicting the fact that no such improvements are found by Algorithm~\ref{alg:1}.
\end{proof}

Observation~\ref{obs:final-set-L} implies that Algorithm~\ref{alg:1} terminates after $L$ iterations. Therefore, in the rest of our analysis we use $A_{\leq L}$ to denote the output set of Algorithm~\ref{alg:1}. Note that the condition on Line~\ref{li:while} of the algorithm implies that the set $A_{\leq i} = A_{\leq i-1} \cup A_i$ finalized at the end of each iteration $i$ is locally optimal with respect to $(m_i,\eps)$-improvements with $N \subseteq A_i$. That is, for every $i \in [L]$, there is no $(m_i, \eps)$ improvement $(S,N)$ for $A_{\leq i}$ with such $N$. In particular, considering $(m_i,\eps)$-improvements of the first kind gives the following observation.
\begin{observation}
\label{obs:local-optimaly}
For every $i \in [L]$, the inequality $f(e \mid A_{\leq i}) < m_i$ holds for every $e \in E \setminus A_{\leq i}$ that obeys $A_{\leq i} + e \in \cI$ at the end of iteration $i$.
\end{observation}
\begin{proof}
Suppose $A_{\leq i}+e \in \cI$ for some $e \in E \setminus A_{\leq i}$ at the end of iteration $i$. If $f(e \mid A_{\leq i}) \geq m_i$, then $(
\{e\}, \emptyset)$ would be a valid $(m_i, \eps)$-improvement, contradicting the fact that the loop on Line~\ref{li:while} of Algorithm~\ref{alg:1} terminated.
\end{proof}

Before concluding this section, let us intuitively explain two aspects of Algorithm~\ref{alg:1}. The first aspect is the choice of powers of $2$ to define the thresholds $m_i$. This choice guarantees that it is always beneficial (or at least non-detrimental) to replace a single element of $A_i$ with two, which is a property that naturally holds in the unweighted version of the problem, and is necessarily for borrowing arguments from this version. The second aspect is that unlike many local-search-based algorithms for submodular maximization, Algorithm~\ref{alg:1} considers a very restricted set of possible improvements, namely, improvements that only involve the removal of up to a single element and the addition of up to two elements. The restriction to add at most two elements is a natural consequence of the restriction to remove at most a single element because, as discussed above, if removing a single element enables the addition of two or more elements, then this is already beneficial (or at least non-detrimental) to add two of them. Thus, it only remains to understand why Algorithm~\ref{alg:1} only considers improvements removing at most a single element. The reason for that is that the exchange property that we use (Lemma~\ref{lem:weights-main}) only guarantees the existence of such improvements, and therefore, taking advantage in the analysis of improvements removing multiple elements would require proving a stronger exchange property.


\section{Analysis of the Main Algorithm}
\label{sec:analysis}

In this section, we bound the approximation guarantee of Algorithm~\ref{alg:1} when $f\colon 2^E \to \bR$ is a general submodular function with $f(\emptyset) \geq 0$. We then use this bound to prove Theorem~\ref{thm:monotone_result}. Fix a matroid $k$-parity constraint $(E, \cI)$. In the following, we consider an execution of Algorithm~\ref{alg:1} on this matroid $k$-parity constraint and the above function $f$. 
Recall that $A_{\leq L}$ is the output of this algorithm, and thus our objective is to lower bound $f(A_{\leq L})$ in terms of $f(O)$, where $O$ is an optimal solution for our given instance (i.e., a set in $\cI$ maximizing $f$). Without loss of generality, we may restrict our analysis to optimal solutions that satisfy the following property.
\begin{definition}
A feasible set $O$ is \emph{strictly down-monotone (with respect to $f$)} if $f(x \mid O - x) > 0$ for every $x \in O$.
\end{definition}
To see that such a restriction is without loss of generality, note that as long as $f(x \mid O - x) \leq 0$ for some element $x \in O$, removing this element from $O$ gives a new feasible solution that is smaller, but has at least the same value as $O$ and thus is also optimal. Hence, any minimal size optimal solution is strictly down-monotone.

Our analysis below does not use any property of $O$ beside its strict down-monotonicity. In particular, it does not use the optimality of $O$, and therefore, applies also when $O$ is a strictly down-monotone set that is not necessarily optimal. This property of the analysis is used in Section~\ref{sec:non-monotone}. 
The analysis consists of a few steps. First, in Section~\ref{sec:part-elem-o}, we partition the elements of $O$ into sets $\{O_i \mid i \in [L]\}$. Intuitively, each $O_i$ is a subset of $O$ whose contribution is bounded in terms of the contribution of the elements of $A_i$, i.e., the elements selected in iteration $i$ of Algorithm~\ref{alg:1}. We quantify these contributions in Section~\ref{sec:defin-weights-charg} by defining weights $w$ and $\ow$ for the elements of $A_{\leq L}$ and $O$, respectively. Section~\ref{sec:setting-up-charges} then classifies the elements of $O$ into three different types of elements, and Section~\ref{sec:main-charging-scheme} separately bounds the total weight of each type in terms of the weight of $A_{\leq L}$ via appropriate charging arguments. 

In the linear case, the weights $w$ and $\ow$ of elements are identical and independent of the sets Algorithm~\ref{alg:1} constructs. This makes our charging scheme essentially the same as the one used in~\cite{singer2025better}. However, in the submodular case, the weights depend on the final set $A_{\leq L}$ returned by the algorithm. In order to analyze the effect of the random shift $\tau$ selected by Algorithm~\ref{alg:1}, we must thus introduce, in Section~\ref{sec:defin-auxil-weights}, a set of auxiliary weights $u$ quantifying the value of $O$ in a fashion that is independent of $A_{\leq L}$. Our analysis then explicitly considers how these weights $u$ differ from $\ow$. Specifically, we show that if there is a large discrepancy between $u$ and $\ow$, then we can handle the resulting loss in the analysis by balancing with an alternative lower bound on $f(A_{\leq L})$ that gains from such a discrepancy.

\subsection{Partitioning the Elements of \texorpdfstring{$O$}{O}}
\label{sec:part-elem-o}

This section describes a partition of $O$ into disjoint sets $O_1,O_2,\dotsc,O_{L}$. Analogously to the sets $A_{\leq i}$, for every integer $0 \leq i \leq L$, we use $O_{\leq i}$ to denote the union $\cup_{j = 1}^{i} O_j$. Intuitively, each set $O_i$ represents elements from $O$ whose values can be bounded in terms of the elements selected as $A_i$ in iteration $i$ of Algorithm~\ref{alg:1} and thus need not be considered when later iterations of this algorithm are analyzed. In line with this intuition, we prove below that the elements of $O \setminus O_{\leq i}$, whose values needs to be bounded in terms of elements selected by later iterations, can all be feasibly added to the solution $A_{\leq i}$ that Algorithm~\ref{alg:1} has after iteration $i$.

We can now formally define the partition of $O$. Recall that $O_{\leq 0} = \emptyset$ by definition. For every $i \in [L]$, the set $O_i$ is given by the recursive formula
\[
	O_i
	\triangleq
	\{o \in O \setminus O_{\leq i - 1} \mid N^{(i)}_o \neq \emptyset\}
	\enspace,
\]
where $\{N_o^{(i)} \subseteq A_{\leq i} \mid o \in A_{\leq i - 1} \cup (O \setminus O_{\leq {i - 1}})\}$ is the collection of sets obtained by applying Lemma~\ref{lem:weights-main} with the sets $A_{\leq i}$ and $A_{\leq i - 1} \cup (O \setminus O_{\leq i - 1})$ as the sets $A$ and $B$ of the lemma, respectively. Recall that applying Lemma~\ref{lem:weights-main} with these sets requires them to be feasible. The feasibility of $A_{\leq i}$ is clear. The feasibility of $A_{\leq i - 1} \cup (O \setminus O_{\leq i - 1})$ follows from the next observation.
\begin{observation} \label{obs:construction_properties}
For every integer $0 \leq i \leq L$, $A_{\leq i} \cup (O \setminus O_{\leq i})$ is feasible and $A_{\leq i} \cap (O \setminus O_{\leq i}) = \emptyset$.
\end{observation}
\begin{proof}
We prove the observation by induction on $i$. First, note that the observation is trivial for $i = 0$ because $A_{\leq 0} \cup (O \setminus O_{\leq 0}) = O$ and $A_{\leq 0} = \emptyset$. Assume now that the observation holds for every integer $0 \leq i' < i$, and let us prove it for $i$.

The induction hypothesis implies that $A_{\leq i'} \cup (O \setminus O_{\leq i'})$ is feasible for every $0 \leq i' < i$, and therefore, the consruction of $O_i$ via the above process is well-defined. The first part of the observation now follows from claim~\ref{item:emptyN} of Lemma~\ref{lem:weights-main} because $O \setminus O_{\leq i} = \{o \in O \setminus O_{\leq i - 1} \mid N^{(i)}_o = \emptyset\}$, and the second part of the observation holds because claim~\ref{item:intersection} of Lemma~\ref{lem:weights-main} guarantees that every element $o$ that belongs both to $A_{\leq i}$ and $O \setminus O_{\leq i - 1}$ has $N^{(i)}_o = \{o\}$, and therefore, belongs to $O_i$.
\end{proof}

\subsection{Defining Weights for the Charging Argument}
\label{sec:defin-weights-charg}
We now define the weights we use to quantify the values of elements in $O$ and $A_{\leq L}$. An element's weight is based on the \emph{incremental value} it adds to a given set, and in the linear case it corresponds exactly to the weight of this element in the definition the objective function $f$ (as long as this weight is non-negative). In order to define the incremental values, we order the elements of $A_{\leq L}$ as $a_1,a_2, \dotsc,a_{|A_{\leq L}|}$, in the order in which they were last added by Algorithm~\ref{alg:1} to any of the sets $A_1, A_2, \dotsc, A_L$. Applying $(\theta,\eps)$-improvements of the third kind involves adding $2$ elements to $A_i$ at the same time, which creates an ambiguity. To resolve the ambiguity, recall that in such improvements, Definition~\ref{def:improvement} specifies that there must be a way to assign the two added elements the labels $x_1$ and $x_2$ so that $f(x_1 \mid A) \geq \theta$ and $f(x_2 \mid A + x_1) \geq \theta$, where $A$ is the set improved by the $(\theta, \eps)$-improvement. In our ordering of $A_{\leq L}$, we suppose that the element corresponding to $x_1$ is added to $A_i$ immediately before the element corresponding to $x_2$.\footnote{If both ways to label the added elements obey the above inequalities, then we arbitrary select one of them.} Given our ordering $a_1,\dotsc,a_{|A_{\leq L}|}$, we define the weight of each element $a_j \in A_{\leq L}$ as
\begin{equation*}
w(a_j) \triangleq f(a_j \mid \{a_1,a_2,\dotsc,a_{j-1}\})\enspace.
\end{equation*}

We proceed similarly for elements in $O$. Let us denote the elements of $O$ by $o_1,o_2\dotsc,o_{|O|}$ in an arbitrary order.\footnote{While the order $o_1,o_2\dotsc,o_{|O|}$ is arbitrary, it is essential that this order is not chosen in a way that depends on the random bits of Algorithm~\ref{alg:1}.} Then, for any $o_j \in O$, we define
\begin{equation*}
\ow(o_j) \triangleq \max\{0, f(o_j \mid (A_{\leq L} - o_j) \cup \{o_1,o_2,\dotsc,o_{j-1}\})\}
\enspace.
\end{equation*}
Note that we ensure that the $\ow$ weights are always non-negative. Let us briefly provide some intuition for the second expression in the maximum. For any $o_j \in O \setminus A_{\leq L}$, this term is equal to $f(o_j \mid A_{\leq L} \cup \{o_1,\dotsc,o_{j-1}\})$, and thus the weights $\ow$ capture the incremental values obtained when adding all such elements sequentially to the set $A_{\leq L}$. For any element $o_j \in O \cap A_{\leq L}$, we could define $\ow(o_j) = 0$, since it provides no such incremental value. However, it is useful in the linear case to ensure that all the weights we are defining are equal to the natural notion of an element's weight (i.e., $f(\{x\})$) whenever this weight is non-negative. The chosen expression ensures that this is the case even for $x \in O \cap A_{\leq L}$.

It is convenient to extend the above-defined weights to sets. We do this in the natural way, i.e., for any set $T \subseteq A_{\leq L}$, we define $w(T) \triangleq \sum_{e \in T}w(e)$, and for any set $T \subseteq O$, we define $\ow(T) \triangleq \sum_{e \in T} \ow(e)$.
The following lemma proves some basic properties of the weights $w$ and $\ow$. The first three claims of the lemma show that an element's weight must fall into a particular range defined by the values $m_i$, depending on which (if any) set $A_i$ or $O_i$ this element appears in. The final claim relates the two weights $w$ and $\ow$ given to elements that appear in both $A_{\leq L}$ and $O$.


\begin{lemma} \label{lem:weights_properties}
The weights $w$ and $\ow$ satisfy the following.
\begin{compactenum}
\item For every $i \in [L]$ and element $a \in A_i$, $0 < m_i \leq w(a) \leq m_{i-1}$.
\label{item:w(a)_bound}
\item For every $i \in [L]$ and element $o \in O_i$, $\ow(o) \leq m_{i-1}$.\label{item:w(o)_bound}
\item For all $o \in O \setminus O_{\leq L}$, $\ow(o) = 0$.\label{item:o_L+1_bound}
\item For every element $o \in O \cap A_{\leq L}$, $\ow(o) \leq w(o)$.\label{item:intersection-w-bound}
\end{compactenum}
\end{lemma}
\begin{proof}
We prove each claim in turn.
\paragraph{Claim~\ref*{item:w(a)_bound}.}
Consider the moment that $a \in A_i$ was last added to $A_i$ by Algorithm~\ref{alg:1}, and let $\bar{A}_i$ denote the value of this set just before the element $a$ was added by some $(m_i,\eps)$-improvement (if the $(m_i,\eps)$-improvement adds two elements, and $a$ is the one labeled $x_2$ in this improvement, then $\bar{A}_i$ is assumed to already contain the element labeled as $x_1$ by the improvement). Since $a$ is added by an $(m_i,\eps)$-improvement for the set $\bar{A}_i \cup A_{\leq i - 1}$,
\[
	m_i
	\leq
	f(a \mid \bar{A}_i \cup A_{\leq i - 1})
	\leq
	w(a)
	\enspace,
\]
where the second inequality follows from the submodularity of $f$ because all the elements of $A_{\leq L}$ that were last added to the sets $A_1, A_2, \dotsc, A_L$ before the last time in which $a$ was added to $A_i$ must also belong to $\bar{A}_i \cup A_{\leq i - 1}$.

For the upper bound of the claim, first suppose that $a \in A_1$. Then, 
\begin{equation*}
w(a) \leq f(a \mid \emptyset) \leq W \leq W\tau = m_0
\end{equation*}
by submodularity and the definition of $W$. In the general case, suppose $a \in A_i$ for some $i > 1$. Then, $a \in E \setminus A_{\leq i-1}$ with $\bar{A}_i \cup A_{\leq i - 1} + a \in \cI$. Since $\cI$ is down-closed, $A_{\leq i-1} + a \in \cI$. Observation~\ref{obs:local-optimaly} then implies that $m_{i-1} > f(a \mid A_{\leq i-1}) \geq w(a)$, where the final inequality holds by submodularity since all the elements of $A_{\leq i-1}$ were added to the algorithm's solution before $a$ was added to $A_i$.

\paragraph{Claim~\ref*{item:w(o)_bound}.} We have $o \in O_i \subseteq O \setminus O_{\leq i - 1}$. By Observation~\ref{obs:construction_properties}, $(O \setminus O_{\leq i-1}) \cap A_{\leq i-1} = \emptyset$ and $(O \setminus O_{\leq i-1}) \cup A_{\leq i-1} \in \cI$; and thus, $o \in E \setminus A_{\leq i-1}$ with $A_{\leq i-1} + o \in \cI$. If $i > 1$, then by Observation~\ref{obs:local-optimaly}, we must then have $f(o \mid A_{\leq i-1}) < m_{i-1}$. If $i = 1$, we get the slightly weaker inequality $f(o \mid A_{\leq i-1}) \leq f(o \mid \varnothing) \leq W \leq m_0 = m_{i - 1}$ by the submodularity of $f$. In either case, since $m_{i-1}$ is positive, $\max\{0, f(o \mid A_{\leq i - 1})\} \leq m_{i - 1}$. Suppose that $o$ is the $j$-th element $o_j$ in our ordering of $O$. Then, by submodularity,
\begin{align*}
	\ow(o)
	=
	\ow(o_j) ={} & \max\{0, f(o_j \mid (A_{\leq L} - o_j) \cup \{o_1,o_2,\dotsc,o_{j-1}\})\} \\\leq{} & \max\{0, f(o_{j} \mid A_{\leq {i-1}})\} \leq m_{i - 1}\enspace.
\end{align*}

\paragraph{Claim~\ref*{item:o_L+1_bound}.} By Observation~\ref{obs:construction_properties}, $O \setminus O_{\leq L}$ and $A_{\leq L}$ are disjoint and $(O \setminus O_{\leq L}) \cup A_{\leq L} \in \cI$. Thus, for every $o_j \in O \setminus O_{\leq L}$ we have $o_j \in E \setminus A_{\leq L}$ and $A_{\leq L} + o_j \in \cI$. Observation~\ref{obs:final-set-L} then implies that $f(o_j \mid A_{\leq L}) \leq 0$. Thus, by the submodularity of $f$,
\begin{equation*}
\ow(o_j) = \max\{0, f(o_j \mid (A_{\leq L} - o_j) \cup \{o_1,o_2,\dotsc,o_{j - 1}\})\} \leq \max\{0, f(o_j \mid A_{\leq L})\} = 0\enspace.
\end{equation*}

\paragraph{Claim~\ref*{item:intersection-w-bound}}
Suppose $o = o_j \in A_{\leq L} \cap O$. Then, we must have $o = a_{j'}$ for some $a_{j'} \in A_{\leq L}$. The first claim of the present lemma yields $w(a_{j'}) > 0$ since $a_{j'}$ must appear in some set $A_i$. Thus, by submodularity,
\begin{align*}
w(o)
= 
w(a_{j'}) 
= 
\max\{0, w(a_{j'})\}
&=
\max\{0, f(a_{j'} \mid \{a_1,\dotsc,a_{j'-1}\})\} \\
&\geq 
\max\{0, f(a_{j'} \mid (A_{\leq L} - a_{j'}) \cup \{o_1,\dotsc,o_{j-1}\} )\} \\
&=
\max\{0, f(o_j \mid (A_{\leq L} - o_j) \cup \{o_1,\dotsc,o_{j-1}\})\}
= \ow(o)\enspace.\qedhere
\end{align*}
\end{proof}

\subsection{Classifying the Elements of \texorpdfstring{$O$}{O} for Charging}
\label{sec:setting-up-charges}
In our charging scheme, we consider three different types of elements of $O$. The first is the elements that belong to some $O_i$, have weights approximately equal to the weights of the elements in $A_i$ (up to a factor of $m_{i-1}/m_i = 2$) and can charge to only a single element of $A_i$. We denote below by $O^{(s)}$ the set of these elements. Intuitively, our analysis uses an argument borrowed from the $(k/2 + \eps)$-approximation algorithm for the unweighted matroid $k$-parity problem~\cite{lee2013matroid} to show that the elements of $O^{(s)}$ only contribute an additive constant to the approximation guarantee. The next type is the elements of $O \setminus O_{\leq L}$, which are ``left over'' at the end of Algorithm~\ref{alg:1}. The contributions of these elements are small, and can be essentially ignored. The final type is all remaining elements, i.e., $O_{\leq L} \setminus O^{(s)}$.

To formally define the set $O^{(s)}$ of the first type of elements of $O$, we define a set $N_o$ for every element $o \in O_{\leq L}$ as follows. For any $o \in O_{\leq L}$, let $i$ be the single integer in $[L]$ for which $o \in O_i$, and define $N_o \triangleq N^{(i)}_o$ (recall that $N^{(i)}_o$ is defined in Section~\ref{sec:part-elem-o}). Intuitively, one can think of $N_o$ as the set of elements of $A_{\leq L}$ that are responsible for the inclusion of $o$ in $O_i$. We then define for every $i \in [L]$ a subset $O_i^{(s)} \subseteq O_i$ given by
\begin{align*}
	O_i^{(s)} &\triangleq \{o \in O_i \mid \ow(o) > m_i \text{ and } |N_o| = 1\}\enspace.
\end{align*}
Then, we can define $O^{(s)} \triangleq \bigcup_{i = 1}^L O_i^{(s)}$. Notice that this definition is consistent with the above intuition in the following senses. First, each element $o \in O^{(s)}$ has indeed only a single element in $N_o$ to ``blame'' for its inclusion in $O_i$ and thus only a single element to charge to. Second, the weight of $o$ is comparable to the weight of this single element because $N_o \subseteq A_i$ (see Observation~\ref{obs:N_o_inclusion} below), which implies by claim~\ref{item:w(a)_bound} of Lemma~\ref{lem:weights_properties} that the weight of every element in $N_o$ is in the range $[m_{i}, m_{i-1}]$.

\begin{observation} \label{obs:N_o_inclusion}
For every $i \in [L]$, $N_o = N^{(i)}_o \subseteq A_i$ for every $o \in O_i$.
\end{observation}
\begin{proof}
Claim~\ref{item:intersection} of Lemma~\ref{lem:weights-main} guarantees that $N^{(i)}_o$ can intersect $A_{\leq i - 1}$ only on the element $o$ itself. However, $o \in O_i$ and thus does not belong to $A_{\leq i - 1}$ by Observation~\ref{obs:construction_properties}.
\end{proof}

The following lemma states some properties of the sets defined above. The first two claims of the lemma show that every element $o \in O^{(s)}$ ``blames'' a different element of $A_{\leq L}$ for its inclusion in $O_{\leq L}$, and furthermore, they show that this element is almost as valuable as $o$.  Claim~\ref{item:at_most_k_N_o} of the lemma shows that every element of $A_{\leq L}$ is ``blamed'' by at most $k$ elements of $O_{\leq L}$.

\begin{samepage}
\begin{lemma} \label{lem:set_weights_properties} The above defined sets obey the following. 
\begin{compactenum}
\item For every $o \in O^{(s)}$, $\ow(o) \leq (1+\eps) \cdot w(N_o)$.\label{item:O_i^s}
\item For every two elements $o, o' \in O^{(s)}$, $N_o \cap N_{o'} = \emptyset$.\label{item:empty_intersection}
\item Every $a \in A_{\leq L}$ appears in $N_o$ for at most $k$ distinct $o \in O_{\leq L}$.\label{item:at_most_k_N_o}
\end{compactenum}
\end{lemma}
\end{samepage}
\begin{proof}
First, we show that all $o \in O_i^{(s)} \setminus A_{\leq i}$ must have relatively large marginal contribution with respect to the set $A_{\leq i}$. By definition, any $o_j \in O_i^{(s)}$ has $0 < m_i < \ow(o_j) = f(o_j \mid (A_{\leq L} - o_j) \cup \{o_1,\dotsc,o_{j-1}\})$. If $o_j \not\in A_{\leq i}$, then $(A_{\leq L} - o_j) \cup \{o_1,\dotsc,o_{j-1}\}) \supseteq A_{\leq i} \cup O'$ for any $O' \subseteq \{o_1,\dotsc,o_{j-1}\}$ and thus, by submodularity,
\begin{equation}
\label{eq:ow-1}
m_i < \ow(o_j) \leq f(o_j \mid A_{\leq i})\,,
\end{equation}
and also
\begin{equation}
\label{eq:ow-2}
m_i < \ow(o_j) \leq f(o_j \mid A_{\leq i} + o_{j'})\,,
\end{equation}
for any $j' < j$. Intuitively, these inequalities show that any such $o_j$ has marginal contribution large enough to be part of an $(m_i, \eps)$-exchange at the end of the iteration $i$. Using these observations, we now prove each claim of the lemma in turn.
\paragraph{Claim~\ref*{item:O_i^s}.} 
Consider any $o \in O^{(s)}$ and let $i \in [L]$ be the (unique) value such that $o \in O^{(s)}_i$. Then, since $o \in O_i \subseteq O \setminus O_{\leq i-1}$, we must have $o \not \in A_{\leq i - 1}$ by Observation~\ref{obs:construction_properties}. If $o \in A_i \cap O_i^{(s)} \subseteq A_i \cap O_i$, then by claim~\ref{item:intersection} of Lemma~\ref{lem:weights-main}, we are guaranteed that $N_o = \{o\}$, which implies $w(o) = w(N_o)$, and claim \ref{item:intersection-w-bound} of Lemma~\ref{lem:weights_properties} further gives $\ow(o) \leq w(o)$. Thus, $\ow(o) \leq (1 + \eps) \cdot \ow(o) \leq (1 + \eps) \cdot w(o) = (1 + \eps) \cdot w(N_o)$, as required.

Now, suppose that $o \in O_i^{(s)} \setminus A_i$. Then, $o \not \in A_{\leq i}$, and $N_o = N_o^{(i)} = \{a\}$ for some element $a \in A_i$. By the properties of $N_o^{(i)}$ (specifically, Lemma~\ref{lem:weights-main}, claim~\ref{item:singletonN}), we have $A_{\leq i} - a + o \in \cI$. As shown above in \eqref{eq:ow-1}, $f(o \mid A_{\leq i}) > m_i$. Thus, $(\{o\},\{a\})$ satisfies all the properties required of an $(m_i,\eps)$ improvement of the second type, except for the inequality $f(A_{\leq i} - a + o) - f(A_{\leq i}) \geq \eps m_i$. Since Algorithm~\ref{alg:1} terminated iteration $i$, and finalized the set $A_{\leq i}$, without applying this improvement, it must have been the case that this inequality was not satisfied, i.e.,
\begin{align*}
  \eps m_i 
  >{} &
  f(A_{\leq i} - a + o) - f(A_{\leq i})
	=
	f(o \mid A_{\leq i} - a) - f(a \mid A_{\leq i} - a)\\
  \geq{} &
  f(o \mid A_{\leq i}) - f(a \mid A_{\leq i} - a) 
	\geq
	\ow(o) - f(a \mid A_{\leq i} - a) 
  \geq 
  \ow(o) - w(a)\enspace,
\end{align*}
where the second inequality follows from submodularity, the third inequality holds by Inequality~\eqref{eq:ow-1}, and the last inequality follows from submodularity and the definition of $w$. Rearranging, we have
\begin{equation*}
  \ow(o) 
  \leq 
  w(a) + \eps m_i  
  \leq 
  (1+\eps) \cdot w(a) 
  = 
  (1+\eps) \cdot w(N_o)\enspace,
\end{equation*}
where the second inequality follows from claim \ref{item:w(a)_bound} of Lemma~\ref{lem:weights_properties} since $a \in A_i$.

\paragraph{Claim~\ref*{item:empty_intersection}.}
Consider any distinct $o_{j},o_{j'} \in O^{(s)}$, where $j' < j$. We show below that if $N_{o_{j'}} \cap N_{o_{j}} \neq \emptyset$, then an $(m_i,\eps)$-improvement of the third form exists for some $A_{\leq i}$, contradicting the fact that iteration $i$ of Algorithm~\ref{alg:1} terminated and finalized the set $A_{\leq i}$ without applying this improvement. 

To this end, suppose that $N_{o_{j'}} \cap N_{o_{j}} \neq \emptyset$. Since both $o_{j'}$ and $o_{j}$ are in $O^{(s)}$, we must have $N_{o_{j'}} = N_{o_{j}} = \{a\}$ for some single element $a$, and this element $a$ must belong to exactly one set $A_i$ (for $i \in [L]$) since the sets $A_i$ are disjoint by construction. Moreover, since $N_o^{(i')} \subseteq A_{i'}$ for all $i' \in [L]$, we further must have $N_{o_j} = N^{(i)}_{o_j} = \{a\}$, $N_{o_{j'}} = N_{o_{j'}}^{(i)} = \{a\}$, and $o_j, o_{j'} \in O^{(s)}_i$ for the same $i \in [L]$. Recall that the sets $\{N_o^{(i)} \mid o \in A_{\leq i-1} \cup (O \setminus O_{\leq i-1})\}$ were constructed to satisfy the properties of Lemma~\ref{lem:weights-main} with respect to $A_{\leq i}$ and $A_{\leq i - 1} \cup (O \setminus O_{\leq i-1})$. Since $N_{o_{j'}}^{(i)} = N_{o_{j}}^{(i)} = \{a\}$, the third such property implies that $(A_{\leq i} - a) \cup \{o_{j'}, o_{j}\} \in \cI$. Additionally, since $o_{j'} \neq o_{j}$ but $N_{o_{j'}}^{(i)} = N_{o_{j}}^{(i)}$, the first such property implies that $o_{j'}, o_j \in (A_{\leq i - 1} \cup (O \setminus O_{\leq i-1})) \setminus A_{\leq i}$. Hence, $o_{j'},o_{j} \in O_i^{(s)} \setminus A_{\leq i}$ and by \eqref{eq:ow-1} and \eqref{eq:ow-2}, respectively, they must obey $f(o_{j'} \mid A_{\leq i}) \geq m_i$ and $f(o_{j} \mid A_{\leq i} + o_{j'}) \geq m_i$. Thus, $(\{o_{j'}, o_{j}\}, \{a\})$ is a valid $(m_i, \eps)$-improvement for $A_{\leq i}$, yielding the desired contradiction.

\paragraph{Claim~\ref*{item:at_most_k_N_o}.} Fix any $a \in A_{\leq L}$ and let $O' = \{o \in O_{\leq L} \mid a \in N_o\}$. By the same argument as in the previous claim, since the sets $A_i$ are disjoint and $N_o = N_o^{(i)} \subseteq A_i$ for all $o \in O_i$, there must be some single $i \in [L]$ such that $a \in A_i$ and $N_o = N_o^{(i)}$ for all $o \in O'$. For this $i$, the sets $\{N_o^{(i)} \mid o \in A_{\leq i-1} \cup (O \setminus O_{\leq i-1})\}$ were constructed to satisfy the properties of Lemma~\ref{lem:weights-main}. By the fourth such property, $a$ appears in at most $k$ of them, and hence, $|O'| \leq k$.
\end{proof}

\subsection{Defining Auxiliary Weights Independent of \texorpdfstring{$A_{\leq L}$}{A\_\{≤L\}}}
\label{sec:defin-auxil-weights}

The weights, sets and properties given in the preceding subsections are sufficient to carry out a charging argument analogous to the one presented for the linear case in~\cite{singer2025better}. As in their proof, this charging argument would use the random shift $\tau$ to argue that, in expectation, an element $o \in O_{\leq L} \setminus O^{(s)}$ has value significantly less than all elements $a \in N_o$. The main difficulty in extending the argument to the submodular case is that the weights $w$ and $\ow$ depend crucially on the set $A_{\leq L}$ constructed by the algorithm, and hence, also on the random choice of $\tau$. To circumvent this difficulty, we consider the incremental contribution of each element $o \in O$ to $f(O)$ alone. Formally, we order the elements $o_1,\dotsc,o_{|O|}$ of $O$ as in the definition of $\ow$, and then, for every $o_j \in O$, define the auxiliary weight
\begin{equation*}
u(o_j) \triangleq f(o_j \mid \{o_1,\dotsc,o_{j-1}\})\enspace.
\end{equation*}
Note that $\sum_{o \in O}u(o) = f(O) - f(\emptyset)$. Moreover, the auxiliary weights $u$ are independent of the execution of Algorithm~\ref{alg:1}, and in particular, the random choice of $\tau$. The following observation relates the weights $\ow$ and $u$.
\begin{observation} \label{obs:positive_weights_diff}
For every $o_j \in O$, $u(o_j) > 0$ and $u(o_j) \geq \ow(o_j)$. When $f$ is linear, the last inequality holds as an equality.
\end{observation}
\begin{proof}
Consider any $o_j \in O$. Since we assume that $O$ is strictly down-monotone, we must have $f(o_j \mid O - o_j) > 0$ for all $o_j \in O$. The submodularity of $f$ then implies that $0 < f(o_j \mid O - o_j) \leq f(o_j \mid \{o_1,\dotsc,o_{j-1}\}) = u(o_j)$. Thus, $u(o_j) > 0$, and by using submodularity again, this further implies that
\begin{align*}
	u(o_j)
	={} 
	\max\{0, u(o_j)\}
	&=
	\max\{0, f(o_j \mid \{o_1,\dotsc,o_{j-1}\})\} \\
	&\geq{} 
	\max\{0, f(o_j \mid (A_{\leq L} - o_j) \cup \{o_1,\dotsc, o_{j-1}\}\}
	= \ow(o_j)
	\enspace.
\end{align*}
Furthermore, when $f$ is linear, the inequality in the above derivation holds as an equality since both its sides are equal to $\max\{0, f(\{o_j\})\}$.
\end{proof}

For general submodular functions, we may have a strict inequality $u(o) > \ow(o)$ for some $o \in O$. In the following section, this potential discrepancy results in a loss proportional to $u(o) - \ow(o)$ in our charging analysis. However, the following lemma shows that this discrepancy can also be used to derive an alternative lower bound on $f(A_{\leq L})$. To get our results, we balance this lower bound against the bounds obtained in Section~\ref{sec:main-charging-scheme} via our charging scheme.

\begin{lemma}
\label{lem:alternative_bound}
It holds that
\begin{equation*}
    \sum_{o \in O}[u(o) - \ow(o)] \leq f(A_{\leq L}) - f(A_{\leq L} \mid O) 
		\enspace.
  \end{equation*}
\end{lemma}
\begin{proof}
By the definitions of the weights $u(o)$ and $\ow(o)$,
\begin{align*}
    \sum_{o \in O}[u(o) - \ow(o)]
    ={} &
    \sum_{j = 1}^{|O|}f(o_j \mid \{o_1,\dotsc,o_{j-1}\}) - \sum_{o_j \in O} \max\{0, f(o_j \mid (A_{\leq L} - o_j) \cup \{o_1,\dotsc,o_{j-1}\})\}
    \\
    \leq {} &
		\sum_{j = 1}^{|O|}f(o_j \mid \{o_1,\dotsc,o_{j-1}\}) - \sum_{o_j \in O \setminus A_{\leq L}} \mspace{-18mu} f(o_j \mid A_{\leq L} \cup \{o_1,\dotsc,o_{j-1}\})\\
		={} &
    \left[f(O) - f(\emptyset)\right] - \left[f(O \cup A_{\leq L}) - f(A_{\leq L})\right]\\
		\leq{} &
		f(A_{\leq L}) - f(A_{\leq L} \mid O) 
		\enspace.
		\qedhere
\end{align*}
\end{proof}

\subsection{The Charging Scheme}
\label{sec:main-charging-scheme}

Recall that Section~\ref{sec:setting-up-charges} classified the elements of $O$ into three types: the elements of $O^{(s)}$, the elements of $O \setminus O_{\leq L}$, and the remaining elements (i.e., the elements of $O_{\leq L} \setminus O^{(s)}$). Each one of the first three lemmata of this section upper bounds $u(o)$ for the elements $o \in O$ of one of these types. Since $u(o)$ is the marginal contribution of $o$ to $O$, later in the section we are able to derive a bound on $f(O)$ by combining the upper bounds on $u(o)$ proved by the three lemmata. As noted previously, all our upper bounds on $u(o)$ include a loss term that is proportional to $u(o) - \ow(o)$. This loss term arises due to the possible discrepancy between the weights $\ow$ and the auxiliary weights $u$ we are bounding. To get our final bound on $f(O)$, we need also Lemma~\ref{lem:alternative_bound}, which gives a bound on $f(O)$ that improves as $\sum_{o \in O} [u(o) - \ow(o)]$ grows and thus allows us to cancel the above-mentioned loss terms.

\begin{lemma}
\label{lem:Os-bound}
For each $o \in O^{(s)}$, 
$
u(o) \leq (1+\eps) \cdot w(N_o) + [u(o) - \ow(o)]
$.
Moreover, 
$
\sum_{o \in O^{(s)}} w(N_o) \leq f(A_{\leq L} \mid \emptyset)
$.
\end{lemma}
\begin{proof}
Claim~\ref{item:O_i^s} of Lemma~\ref{lem:set_weights_properties} immediately implies that for any $o \in O^{(s)}$, 
\begin{equation*}
u(o) = \ow(o) + [u(o) - \ow(o)] \leq (1+\eps)\cdot w(N_o) + [u(o) - \ow(o)]\enspace.
\end{equation*}

Claim~\ref{item:empty_intersection} of Lemma~\ref{lem:set_weights_properties} further implies that every element of $A_{\leq L}$ appears in $N_o$ for at most one element $o \in O^{(s)}$. Together with the fact that $w(a) > 0$ for every $a \in A_{\leq L}$ (claim~\ref{item:w(a)_bound} of Lemma~\ref{lem:weights_properties}), this implies that
\begin{equation*}
\sum_{o \in O^{(s)}}w(N_o) \leq \sum_{a \in A_{\leq L}}w(a) = f(A_{\leq L} \mid \emptyset) \enspace.\qedhere
\end{equation*}
\end{proof}

\begin{lemma}
\label{lem:O'-bound}
For each $o \in O \setminus O_{\leq L}$, 
$
u(o) = u(o) - \ow(o)
$.
\end{lemma}
\begin{proof}
This follows immediately from claim~\ref{item:o_L+1_bound} of Lemma~\ref{lem:weights_properties}, which states that $\ow(o) = 0$ for all $o \in O \setminus O_{\leq L}$. 
\end{proof}

Next, we would like to bound $u(o)$ for the elements in $O_{\leq L} \setminus O^{(s)}$. Unlike the bounds in the previous two lemmata, this bound depends on the random shift $\tau$ selected by the algorithm. To formulate this bound, let us define, for each $o \in O$, $\mo \triangleq \min \{ m_i \mid i \geq 0,  m_i \geq u(o) \}$ to be the value of the smallest threshold $m_i$ that is at least $u(o)$ (observe that $\mo$ is well-defined since $u(o) > 0$ by Observation~\ref{obs:positive_weights_diff} and $u(o) \leq W \leq m_0$ by submodularity). Additionally, we define the ratio $r_o \triangleq \mo/u(o)$. Notice that $r_o$ is always in the range $[1, 2)$ since $u(o) \leq m^{(o)}$ and $m_{i-1}/m_i = 2$ for all $i$. Given this notation, we can now state the promised bound as Lemma~\ref{lem:O-other-bound}. This lemma introduces a parameter $d \geq 1$, which is later used to control the balance between the bound of this lemma and the aforementioned bound from Lemma~\ref{lem:alternative_bound}.
\begin{lemma}
  \label{lem:O-other-bound}
  For every $o \in O_{\leq L} \setminus O^{(s)}$ and any $d \geq 1$,
\begin{equation*}
\rho_{o,d} \cdot u(o) \leq w(N_o) + d[u(o) - \ow(o)]\enspace,
\end{equation*}
where
$\rho_{o,d} \triangleq r_o$ when $f$ is linear, and $\rho_{o,d} \triangleq \min\Big\{r_o,   \frac{1-\frac{1}{2}(1-1/d)}{1 - r_o^{-1}(1-1/d)}\Big\}$ when $f$ is a general submodular function.
\end{lemma}
\begin{proof}
We begin the proof by showing that for all $o \in O_{\leq L} \setminus O^{(s)}$, there exist some threshold $m_j$ such that $\ow(o) \leq m_j \leq w(N_o)$. Since $o \in O_{\leq L}$, we must have $o \in O_i$ for some unique $i \in [L]$, and thus $N_o$ is well-defined and obeys $|N_o| \geq 1$. Since $o \not\in O^{(s)}$, this implies that either $\ow(o) \leq m_i$ or $|N_o| \geq 2$.  In the first case,
\begin{equation*}
\ow(o) \leq m_{i} \leq |N_o| \cdot m_i  \leq w(N_o) \enspace,
\end{equation*}
where the last inequality holds since $N_o \subseteq A_i$ and $w(a) \geq m_{i}$ for every $a \in A_i$ (by claim~\ref{item:w(a)_bound} of Lemma~\ref{lem:weights_properties}).
In the second case,
\begin{equation*}
\ow(o) \leq m_{i-1} = 2m_i \leq |N_o|\cdot m_i 
\leq w(N_o) \enspace,
\end{equation*}
where the first inequality holds by claim~\ref{item:w(o)_bound} of Lemma~\ref{lem:weights_properties} because $o \in O_i$, and the last inequality again holds since $N_o \subseteq A_i$ and $w(a) \geq m_{i}$ for every $a \in A_i$ (claim~\ref{item:w(a)_bound} of Lemma~\ref{lem:weights_properties}). Thus, in all cases there is indeed always a threshold $m_j$ with $\ow(o) \leq m_j \leq w(N_o)$. 

We now consider two cases based on the relationship between $u(o)$ and the threshold $m_j$. If $u(o) \leq m_j$, then because $\mo$ is the smallest threshold that is at least $u(o)$, we get
\begin{equation*}
u(o) = \frac{\mo}{r_o} \leq \frac{m_j}{r_o} \leq \frac{w(N_o)}{r_o}\enspace.
\end{equation*}
Rearranging this inequality yields
\begin{equation}
\label{eq:u(o)-bound-1}
r_o \cdot u(o) \leq w(N_o) \leq w(N_o) + d[u(o) - \ow(o)]\enspace,
\end{equation}
where the last inequality follows from $u(o) \geq \ow(o)$ (Observation~\ref{obs:positive_weights_diff}). When the function $f$ is linear, this case is the only possible case because $u(o) = \ow(o) \leq m_j$ for such $f$ (recall that $u(o) > 0$ by Observation~\ref{obs:positive_weights_diff}). This completes the proof of the lemma for linear functions.

For general submodular functions, we may also have $u(o) > m_j$. In this case, $\mo$ must be at least $m_{j - 1}$, which implies $m_j = \frac{1}{2}m_{j - 1} \leq \frac{1}{2}\mo = \frac{u(o)\cdot r_o}{2}$. Thus,
\begin{align*}
\Biggl(1 - \frac{1-1/d}{2} \Biggr) \cdot u(o) &= m_j - \frac{1 - 1/d}{2} \cdot u(o) + [u(o) - m_j] \\
&\leq  m_j - \frac{(1 - 1/d)}{2} \cdot \frac{2m_j}{r_o} + [u(o) - m_j] \\
&=  (1 - r_o^{-1}(1 - 1/d))m_j + [u(o) - m_j] \enspace.
\end{align*}
Note that, since $r_o \geq 1$, it must hold that $1 - r_o^{-1}(1-1/d) \geq 1/d$. Thus, dividing both sides of the above inequality by $1 - r_o^{-1}(1-1/d)$ and simplifying yields
\begin{multline}
\label{eq:u(o)-bound-2}
\frac{1-\frac{1}{2}(1-1/d)}{1 - r_o^{-1}(1-1/d)}
 \cdot u(o) \leq m_j + \left(1 - r_o^{-1}(1-1/d)\right)^{-1}[u(o) - m_j]  \\
\leq m_j + d[u(o) - m_j] 
\leq w(N_o) + d[u(o) - \ow(o)]\enspace,
\end{multline}
where we have used the fact that $\ow(o) \leq m_j \leq w(N_o)$ in the last inequality. The guarantee of the lemma for general submodular functions now follows by taking the smaller of the two lower bounds on $w(N_o) + d[u(o) - \ow(o)]$ given by~\eqref{eq:u(o)-bound-1} and~\eqref{eq:u(o)-bound-2} for the two cases.
\end{proof}

Combining all the above bounds, we obtain the following proposition.
\begin{proposition}\label{prop:main-guarantee-1}
Assuming $f$ is a submodular function such that $f(\emptyset) \geq 0$, the output set $A_{\leq L}$ of Algorithm~\ref{alg:1} obeys, for every two values $d \geq 2$ and $d' \leq k + 1$ and any strictly down-monotone $O \in \cI$,
\[
	(k + 1 + 2\eps) \cdot \bE[f(A_{\leq L})]
	\geq
        d' \cdot f(O) - d\cdot \sum_{o \in O}\bE[u(o) - \ow(o)] - \sum_{o \in O} [d' - \expect[\rho_{o,d}]] \cdot u(o)
	\enspace,
\]
where $\rho_{o,d} \triangleq r_o$ when $f$ is linear and $\rho_{o,d} \triangleq \min\Big\{r_o,   \frac{1-\frac{1}{2}(1-1/d)}{1 - r_o^{-1}(1-1/d)}\Big\}$ when $f$ is a general submodular function.
\end{proposition}
\begin{proof}
Combining $\rho_{o, d}$ times the bounds from Lemmata~\ref{lem:Os-bound} and~\ref{lem:O'-bound} with the bound from Lemma~\ref{lem:O-other-bound}, we obtain
{\allowdisplaybreaks\begin{align*}
&{\sum_{o \in O}\rho_{o,d}\cdot u(o) \leq \sum_{\mathclap{o \in O^{(s)}}} 2 \cdot u(o) + \mspace{10mu}\sum_{\mathclap{o \in O \setminus O_{\leq L}}} 2 \cdot u(o)
+ \mspace{10mu}\sum_{\mathclap{o \in O_{\leq L} \setminus O^{(s)}}} \rho_{o,d} \cdot u(o)} \\
&\leq {2} \cdot \sum_{\mathclap{o \in O^{(s)}}} \bigl((1+\eps)w(N_o) + [u(o) - \ow(o)] \bigr)
+ {2} \cdot \sum_{\mathclap{o \in O \setminus O_{\leq L}}} [u(o) - \ow(o)]
  \mspace{5mu} + \mspace{5mu}\sum_{\mathclap{o \in O_{\leq L} \setminus O^{(s)}}} \bigl(w(N_o) + d[u(o) - \ow(o)]) \\
&\leq \sum_{\mathclap{o \in O_{\leq L}}}w(N_o) + (1 + 2\eps)\sum_{\mathclap{o \in O^{(s)}}}w(N_o) + d \cdot \sum_{o \in O}[u(o) - \ow(o)]\\
&\leq \sum_{\mathclap{a \in A_{\leq L}}}k \cdot w(a) + (1 + 2\eps)\sum_{\mathclap{o \in O^{(s)}}}w(N_o) + d \cdot \sum_{o \in O}[u(o) - \ow(o)] \\
&\leq (k + 1 + 2\eps) \cdot f(A_{\leq L} \mid \emptyset) + d 
\cdot \sum_{o \in O}[u(o) - \ow(o)]
\enspace,
\end{align*}}%
where the first inequality uses the facts that $\rho_{o, d} \leq r_o \leq 2$ and $u(o) > 0$ (by Observation~\ref{obs:positive_weights_diff}); the third inequality follows from $d \geq 2$ and the inequality $u(o) - \ow(o) \geq 0$ (which holds by Observation~\ref{obs:positive_weights_diff}); the fourth inequality follows from 
the fact that each $a \in A_{\leq L}$ appears in at most $k$ sets $N_o$ (claim~\ref{item:at_most_k_N_o} of Lemma~\ref{lem:set_weights_properties}) and
obeys $w(a) \geq 0$ (claim~\ref{item:w(a)_bound} of Lemma~\ref{lem:weights_properties}); and the final inequality follows from the inequality $\sum_{o \in O^{(s)}}w(N_o) \leq f(A_{\leq L}\mid \emptyset)$ proved by Lemma~\ref{lem:Os-bound}.

Taking expectations on both sides, and noting that $f(O)$ and $u(o)$ do not depend on the randomness of the algorithm, we obtain
{\allowdisplaybreaks\begin{align*}
	d' \cdot f(O)
	&=
	{d' \cdot f(\emptyset) + \bE\bigg[\sum_{\mathclap{o \in O}}\rho_{o,d} \cdot u(o) \bigg] + \bE\bigg[\sum_{o \in O} (d' - \rho_{o,d})\cdot u(o)\bigg]}\\
	&\leq
	d' \cdot f(\emptyset) \mspace{-1mu}+\mspace{-1mu} (k \mspace{-1mu}+\mspace{-1mu} 1 \mspace{-1mu}+\mspace{-1mu} 2\eps) \cdot \bE[f(A_{\leq L} \mid \emptyset)] \mspace{-1mu}+\mspace{-1mu} d \cdot \sum_{\mathclap{o \in O}}\bE[u(o) - \ow(o)] \mspace{5mu} + \mspace{5mu}\sum_{o \in O} [d' - \bE[\rho_{o,d}]] \cdot u(o)\\
	&\leq
	(k + 1 + 2\eps) \cdot \bE[f(A_{\leq L})] + d \cdot \sum_{\mathclap{o \in O}}\bE[u(o) - \ow(o)] \mspace{5mu} + \mspace{5mu}\sum_{o \in O} [d' - \bE[\rho_{o,d}]] \cdot u(o)
	\enspace.
\end{align*}}%
Rearranging this inequality gives the claimed result.
\end{proof}

Using Proposition~\ref{prop:main-guarantee-1} we can now prove our main results for linear and monotone submodular functions. We begin with the simpler case of linear functions.
\begin{theorem}\label{thm:linear-guarantee}
When $f$ is a linear function, the approximation guarantee of Algorithm~\ref{alg:1} is at most $(k + 1) \ln 2 + O(\eps)$.
\end{theorem}
\begin{proof}
For linear functions, $u(o) = \ow(o) = f(\{o\})$ for all $o \in O$, and hence, the terms $\bE[u(o) - \ow(o)]$  in Proposition~\ref{prop:main-guarantee-1} all vanish. Then, since $\rho_{o,d} = r_o$ for linear functions, we have
\begin{equation*}
(k +1 + 2\eps)\cdot \bE[f(A_{\leq L})] \geq d' \cdot f(O) - \sum_{o \in O} [d' - \bE[r_{o}]] \cdot u(o)
\enspace.
\end{equation*}
Let $d' = \ln^{-1}2$ (notice that $d' \leq 2 \leq k + 1$). Below, we prove that $\bE[r_o] = d'$ for every $o \in O$, which implies that the last term on the rightmost side of the above inequality is $0$. Rearranging the above inequality then yields
\[
	\bE[f(A_{\leq L})]
	\geq
	\frac{d'}{k + 1 + 2\eps} \cdot f(O)
	=
	\frac{\ln^{-1} 2}{k + 1 + 2\eps} \cdot f(O)
	=
	\frac{f(O)}{(k + 1) \ln 2 + O(\eps)} 
	\enspace.
\]
Recalling that there must exist at least one optimal solution $O$ that is also strictly down-monotone, the theorem follows from this inequality. The rest of the proof is devoted to proving that indeed $\bE[r_o] = d'$.

Recall that for each $o \in O$, $r_o$ is the ratio $m^{(o)}/u(o)$, where $m^{(o)} = \min\{m_i \mid i \geq 0, m_i \geq u(o)\}$ is the lowest value of a threshold $m_i$ that is at least $u(o)$. We also recall that, for each $i$, $m_i = W\tau 2^{-i} = W 2^{-i+\alpha}$, where $\alpha$ is chosen uniformly at random from $(0,1]$. Intuitively, this implies that $r_o = 2^\beta$ for some $\beta \in [0,1)$ that is obtained by cyclically shifting $\alpha$, and is thus, distributed uniformly at random from $[0,1)$. We give a formal proof of this fact in Lemma~\ref{lem:shift} in Appendix~\ref{apx:r_is_shift}. Given this fact,
\[
\bE\left[ r_o\right] 
= \int_{0}^1 2^\beta \,\mathrm{d}\beta
= \frac{2^\beta}{\ln 2}\biggr|_{\beta = 0}^{\beta = 1}
= \ln^{-1} 2
= d'
\enspace.
\qedhere
\]
\end{proof}



The following proposition gives a guarantee for Algorithm~\ref{alg:1} when the objective is a general submodular (rather than linear) function. Our main result for monotone submodular functions readily follows from this proposition (see Corollary~\ref{cor:monotone-submod-guarantee}), but the proposition itself applies also to non-monotone submodular functions, and is used in Setion~\ref{sec:non-monotone} to derive our main result for such functions.

\begin{proposition}\label{prop:main-guarantee}
When $f$ is a non-negative submodular function, the output set $A_{\leq L}$ of Algorithm~\ref{alg:1} obeys, for every $d \geq 2$ and any strictly down-monotone $O \in \cI$,
\[
	(k + d + 1 + 2\eps) \cdot \bE[f(A_{\leq L})]
	\geq
	d' \cdot f(O) + d \cdot \bE [f(A_{\leq L} \mid O)]
	\enspace,
\]
where $d' = \frac{1 - 1/d}{2 \ln 2} + \frac{d + 1}{2d}$.
\end{proposition}
\begin{proof}
Note that $d' = \frac{1 - 1/d}{2 \ln 2} + \frac{d + 1}{2d} \leq \frac{1}{2 \ln 2} + \frac{3d/2}{2d} < 1 + 3/4 < 2$. Thus, Proposition~\ref{prop:main-guarantee-1} 
and Lemma~\ref{lem:alternative_bound} give together
\begin{align*}
(k + 1 + 2\eps) \cdot {}&\bE[f(A_{\leq L})] 
  \geq
    d' \cdot f(O) - d\cdot \bE\Bigl[\sum_{\mathclap{o \in O}}u(o) - \ow(o)\Bigr] -
    \sum_{o \in O} [d' - \bE[\rho_{o,d}]] \cdot u(o) \\
  &\geq 
    d' \cdot f(O) - d\cdot\bE[f(A_{\leq L}) - f(A_{\leq L} \mid O)] -
    \sum_{o \in O} [d' - \bE[\rho_{o,d}]] \cdot u(o)
		\enspace.
\end{align*}
Below we prove that $\bE[\rho_{o,d}] = d'$ for every $o \in O$, which implies that the last term on the rightmost side of the above inequality is $0$. Notice that, given this observation, the proposition follows by using the linearity of expectation and rearranging the resulting inequality. The rest of the proof is devoted to proving that indeed $\bE[\rho_{o,d}] = d'$.

Recall that Lemma~\ref{lem:shift} in Appendix~\ref{apx:r_is_shift} proves that $r_o = 2^\beta$ for $\beta$ that is distributed uniformly at random from $[0,1)$. This property of $r_o$ underlies the calculation below.
\begin{align*}
\bE[\rho_{o,d}] &= \bE\left[ \min\left\{ r_o, \frac{1 - \frac{1}{2}(1 - 1/d)}{1 - r_o^{-1}(1-1/d)} \right\} \right] \\
&= \int_{0}^1 \min\left\{ 2^\beta, \frac{1 - \frac{1}{2}(1 - 1/d)}{1 - 2^{-\beta}(1-1/d)} \right\}\,\mathrm{d}\beta
\\
&= \int_{0}^{\log_2(1+\frac{1}{2}(1-1/d))} 2^\beta\,\mathrm{d}\beta
+ \int_{\log_2(1+\frac{1}{2}(1-1/d))}^{1} \frac{1 - \frac{1}{2}(1 - 1/d)}{1 - 2^{-\beta}(1-1/d)} \mathrm{d}\beta \\
&= \frac{2^\beta}{\ln 2}\biggr|_{\beta = 0}^{\beta = \log_2(1+\frac{1}{2}(1-1/d))}
+
\frac{1-\frac{1}{2}(1-1/d)}{\ln 2}\cdot \ln\left(2^\beta - (1-1/d)\right)\biggr|_{\beta = \log_2(1+\frac{1}{2}(1-1/d))}^{\beta=1} \\
&= \frac{1-1/d}{2\ln 2} 
+ \frac{1 - \frac{1}{2}(1-1/d)}{\ln 2} \cdot \left(\ln(1 + 1/d) - \ln(\tfrac{1}{2}(1+1/d))\right)\\
&= \frac{1-1/d}{2\ln 2} + 1 - \tfrac{1}{2}(1-1/d)
= \frac{1-1/d}{2\ln 2} + \frac{d+1}{2d} = d'\enspace.
\qedhere
\end{align*}
\end{proof}


\begin{corollary}
\label{cor:monotone-submod-guarantee}
When $f$ is a non-negative monotone submodular function, the approximation guarantee of Algorithm~\ref{alg:1} is at most $\frac{2k \ln 2}{1+\ln2} + O(\sqrt{k}) \leq 0.819k + O(\sqrt{k}).$
\end{corollary}
\begin{proof}
In the case that $f$ is monotone, the expression $f(A_{\leq L} \mid O)$ in Proposition~\ref{prop:main-guarantee} is always non-negative. Thus, for any $d \geq 2$, we have
\begin{equation*}
(k + d + 1 + 2\eps) \cdot \bE[f(A_{\leq L})] \geq d' \cdot f(O)
\enspace,
\end{equation*}
where $d' = \frac{1 - 1/d}{2 \ln 2} + \frac{d + 1}{2d}$. Substituting $d'$ into this inequality, we get that the ratio $f(O)/\bE[f(A_{\leq L})]$ is at most
\[
	\frac{k + d + 1 + 2\eps}{\frac{1 - 1/d}{2 \ln 2} + \frac{d + 1}{2d}}
	=
	\frac{k + d + 1 + 2\eps}{\frac{1 + \ln 2}{2 \ln 2} - d^{-1} \cdot \frac{1 - \ln 2}{2 \ln 2}}
	\leq
	\frac{(k + d + 1 + 2\eps)(1 + 2d^{-1} \cdot \frac{1 - \ln 2}{1 + \ln 2})}{\frac{1 + \ln 2}{2 \ln 2}}
	\enspace,
\]
where the inequality uses the fact that for every two positive values $a, b$ obeying $2b \leq a$, we have $\frac{1}{a - b} \leq \frac{1 + 2b/a}{a}$. Selecting $d = 2\sqrt{k}$ (notice that $d \geq 2$) makes the above expression equal to $\frac{2k \ln 2}{1+\ln 2} + O(\sqrt{k})$, which proves the corollary since there must exist some optimal solution $O$ that is also strictly down-monotone.
\end{proof}

In Appendix~\ref{app:time_complexity}, we show that it is possible to simulate Algorithm~\ref{alg:1} in $O(\eps^{-1}|E|^4) = \Poly(|E|,\allowbreak \eps^{-1})$ time. Together with Theorems~\ref{thm:linear-guarantee} and Corollary~\ref{cor:monotone-submod-guarantee}, this immediately gives the following.
\thmMonotoneResult*


\section{Algorithm for Non-monotone Submodular Functions} \label{sec:non-monotone}

In this section, our goal is to prove Theorem~\ref{thm:non-monotone_result}. As mentioned in Section~\ref{ssc:our_results}, we do that using the framework of~\cite{feldman2023how} for adapting algorithms for maximizing non-negative monotone submodular functions to work also for functions that are not necessarily monotone. Such algorithms often have guarantees that depend on $f(B \cup O)$, where $B$ is their output set and $O$ is an optimal solution, which is a good thing when $f$ is monotone (since $f(B \cup O) \geq f(O)$ in this case), but problematic in the non-monotone case. Notice that the guarantee of Algorithm~\ref{alg:1} given by Proposition~\ref{prop:main-guarantee} indeed conforms to this general trend as it involves the term $f(A_\ell \mid O) = f(A_\ell \cup O) - f(O)$.

The framework of~\cite{feldman2023how} (and the earlier similar frameworks of~\cite{gupta2010constrained,mirzasoleiman2016fast}) is based on the observation that the dependence on the term $f(B \cup O)$ is only bad when $f(B \cup O)$ is significantly smaller than $f(O)$, or in other words, when the elements of $B$ significantly ``damage'' the optimal solution. This observation motivates the following meta-algorithm: execute the algorithm for monotone functions to get a solution $B_1$, remove the elements of $B_1$ from the ground set, then execute the algorithm for monotone functions again to get a second solution $B_2$ and then return the better among the two solutions. For the first run of the algorithm, we obtain a guarantee relating $f(B_1)$ to $f(B_1 \cup O)$, and for the second run we obtain a guarantee relating $f(B_2)$ to $f(B_2 \cup (O \setminus B_1))$. Since $B_1$ and $B_2$ are disjoint, submodularity implies that $f(B_1 \cup O) + f(B_2 \cup (O \setminus B_1)) \geq f(O \setminus B_1) + f(B_1 \cup B_2 \cup O)$, and so the better of $B_1 \cup O$ and $B_2 \cup (O \setminus B_1)$ is at least $\frac{1}{2}f(O \setminus B_1)$. If $f(O \setminus B_1)$ is comparable to $f(O)$, then we are done. 
In the remaining case, $f(O \setminus B_1)$ is significantly smaller than $f(O)$, which by submodularity implies that $f(O \cap B_1)$ must be large. To handle this case as well, we can apply an approximation algorithm for \emph{unconstrained} submodular maximization to the set $B_1$ to obtain a third solution $B'_1 \subseteq B_1$ whose value is at least $\frac{1}{2} \cdot f(B_1 \cap O)$. Here, we use the Double Greedy algorithm of Buchbinder et al.~\cite{buchbinder2015tight} for this purpose.\footnote{Formally, Double Greedy gets a non-negative submodular function $f \colon 2^E \to \nnR$ and returns a set $T \subseteq E$ such that $\bE[f(T)] \geq \frac{1}{2} \cdot \max_{S \subseteq E} f(S)$. To employ this algorithm for our purpose, we pass to it the restriction of $f$ to the ground set $B_1$, which guarantees that its output set $B'_1$ is a subset $B_1$ and that $\bE[f(B'_1)] \geq\frac{1}{2} \cdot \max_{S \subseteq B_1} f(S) \geq \frac{1}{2} \cdot f(B_1 \cap O)$.}

To get the best guarantee via this general framework, it is necessary to repeat the above steps multiple times, which generates a series of disjoint solutions $B_1, B_2, \dotsc$ and a series of additional solutions $B'_1 \subseteq B_1, B'_2 \subseteq B_2, \dotsc$. Algorithm~\ref{alg:repeat} implements this in our context. It receives as input a non-negative (not necessarily monotone) submodular function $f\colon 2^E \to \nnR$, a matroid $k$-parity constraint $(E, \cI)$, and a positive integer parameter $\ell$ determining the number of times that the algorithm for montone functions (Algorithm~\ref{alg:1}) is executed. As expected, the $i$-th execution of Algorithm~\ref{alg:1} (by Line~\ref{line:restriction} of Algorithm~\ref{alg:repeat}) is done with respect to a restriction of the input instance to the ground set $E_{i - 1}$ obtained by removing the elements of $\cup_{j = 1}^{i - 1} B_j$ from $E$. We can use Proposition~\ref{prop:main-guarantee} to analyze the guarantee of Algorithm~\ref{alg:1} on this restricted instance because its objective function $f|_{E_i}$ inherits the submodularity of $f$ and the restricted constraint $(E_{i-1},2^{E_{i - 1}} \cap \cI)$ remains a matroid $k$-parity constraint. 

\begin{algorithm}
Let $E_0 \gets E$.\\
\For{$i = 1$ \KwTo $\ell$}
{
Execute Algorithm~\ref{alg:1} on the function $f|_{E_i}$, the constraint $(E_{i-1}, 2^{E_{i - 1}} \cap \cI)$, and an arbitrary constant value $\eps \in (0, 1)$. Let $B_i$ be the output set of this algorithm. \label{line:restriction}
\\
	Execute the Double Greedy algorithm of~\cite{buchbinder2015tight} on the function $f|_{B_i}$, and let $B'_i$ be the output set of this algorithm.\label{line:double_greedy}\\
	Let $E_i \gets E_{i - 1} \setminus B_i$.
}
\Return the set $B^* \in \{B_i \mid i \in [\ell]\} \cup \{ B'_i \mid i \in [\ell]\}$ with maximum value $f(B^*)$.
\caption{\textsc{Repetitions Algorithm} $(E, f, \cI, \ell)$}
\label{alg:repeat}
\end{algorithm}



We begin the formal analysis of Algorithm~\ref{alg:repeat} with the following two lemmata, which lower bound the values of the individual sets $B_i$ and $B'_i$. As in our analysis of Algorithm~\ref{alg:1} in Section~\ref{sec:analysis}, we assume that $O$ is an arbitrary strictly down-monotone feasible solution, and let $B^*$ be the set produced by Algorithm~\ref{alg:repeat}.
\begin{lemma} \label{lem:basic_lower_bound_B}
For every $i \in [\ell]$ and $d \geq 2$,
\begin{align*}
	(k + d + 1 + 2\eps) \cdot \bE[f(B_i)]
	\geq{} &
	d' \cdot \bE[f(O \cap E_{i - 1})] + d \cdot \expect [f(B_i \mid O)] \enspace,
\end{align*}
where $d' = \frac{1 - 1/d}{2 \ln 2} + \frac{d + 1}{2d}$.
\end{lemma}
\begin{proof}
Notice that the set $B_i$ is the output set of the execution of Algorithm~\ref{alg:1} invoked by Line~\ref{line:restriction} of Algorithm~\ref{alg:repeat}. The ground set of the matroid $k$-parity constraint passed to this execution is $E_{i - 1}$, and thus, $B_i \subseteq E_{i-1}$. Conditioned on $E_{i - 1}$, Proposition~\ref{prop:main-guarantee} implies that
\begin{align} \label{eq:sub-opt}
	(k + d + 1 + 2\eps) \cdot \bE[f(B_i)]
	\geq{} &
	d' \cdot f(O') + d \cdot \expect [f(B_i \mid O')]\\\nonumber
	\geq{} &
	d' \cdot f(O') + d \cdot \expect [f(B_i \mid O' \cup (O \setminus E_{i - 1}))]
\end{align}
for every set $O'$ that is a feasible and strictly down-monotone subset of $E_{i - 1}$. The second inequality in Inequality~\eqref{eq:sub-opt} follows from the submodularity of $f$ because $B_i \subseteq E_{i-1}$. Now, consider the set $O \cap E_{i-1} \subseteq O$. For all $x \in O \cap E_{i - 1}$, we have $f(x \mid O \cap E_{i-1} - x) \geq f(x \mid O - x) > 0$ due to the submodularity of $f$ and the strict down-monotonicity of $O$. Thus, $O \cap E_{i-1}$ must also be a strictly down-monotone solution. Letting $O' = O \cap E_{i-1}$ in~\eqref{eq:sub-opt}, and then taking the expectation over the set $E_{i - 1}$, completes the proof of our claim.
\end{proof}

\begin{lemma} \label{lem:basic_lower_bound_B_prime}
For every $i \in [\ell]$, $\bE[f(B'_i)] \geq \frac{1}{2} \cdot \bE[f(O \cap B_i)]$.
\end{lemma}
\begin{proof}
By the properties of the Double Greedy algorithm, conditioned on the set $B_i$, $\bE[f(B'_i)] \geq \frac{1}{2} \cdot f(O \cap B_i)$. The lemma follows by taking the expectation of both sides of this inequality over the set $B_i$.
\end{proof}
Combining the above results, we have the following. Recall that $B^*$ denotes the output set of Algorithm~\ref{alg:repeat}.
\begin{corollary} \label{cor:basic_lower_bounds}
For every $i \in [\ell]$ and $d \geq 2$, 
\[
	(k + d + 2d'(i - 1) + 1 + 2\eps) \cdot \bE[f(B^*)]
	\geq
	d' \cdot f(O) + d \cdot \expect [f(B_i \mid O)]
	\enspace,
\]
where $d' = \frac{1 - 1/d}{2 \ln 2} + \frac{d + 1}{2d}$.
\end{corollary}
\begin{proof}
By the definition of $B^*$,
{\allowdisplaybreaks\begin{align*}
	(k + d + 1 + 2\eps) \cdot \bE[f(B^*)]
	\geq{} &
	(k + d + 1 + 2\eps) \cdot \bE[f(B_i)]\\
	\geq{} &
	d' \cdot \bE[f(O \cap E_{i - 1})] + d \cdot \expect [f(B_i \mid O)]\\
	\geq{} &
	d' \cdot \bE\bigg[f(O) - \sum_{j = 1}^{i - 1} f(O \cap B_j)\bigg] + d \cdot \expect [f(B_i \mid O)]\\
	={} &
	d' \cdot f(O) - d' \cdot \sum_{j = 1}^{i - 1} \bE[f(O \cap B_j)] + d \cdot \expect [f(B_i \mid O)]\\
	\geq{} &
	d' \cdot f(O) - 2d' \cdot \sum_{j = 1}^{i - 1} \bE[f(B'_j)] + d \cdot \expect [f(B_i \mid O)]\\
	\geq{} &
	d' \cdot f(O) - 2d'(i - 1) \cdot \bE[f(B^*)] + d \cdot \expect [f(B_i \mid O)]
	\enspace,
\end{align*}}%
where the second inequality holds by Lemma~\ref{lem:basic_lower_bound_B}, the third inequality follows from the submodularity of $f$ since $O \cap B_1, O \cap B_2, \dotsc, O \cap B_{i - 1}$ and $O \cap E_{i - 1}$ are a disjoint partition of $O$, and the penultimate inequality holds by Lemma~\ref{lem:basic_lower_bound_B_prime}. The corollary now follows by rearranging this inequality.
\end{proof}

We are now ready to prove Theorem~\ref{thm:non-monotone_result}, which we repeat here for convenience.
\thmNonMonotoneResult*
\begin{proof}
Summing the guarantee of Corollary~\ref{cor:basic_lower_bounds} over all $i \in [\ell]$, we get
\begin{align*}
	\ell \cdot (k + d + d'(\ell - 1) + 1 + 2\eps) \cdot \bE[f(B^*)]
	={} &
	\sum_{i = 1}^\ell (k + d + 2d'(i - 1) + 1 + 2\eps) \cdot \bE[f(B^*)]\\
	\geq{} &
	d'\ell \cdot f(O) + d \cdot \expect \bigg[\sum_{i = 1}^\ell f(B_i \mid O)\bigg]\\
	\geq{} &
	d'\ell \cdot f(O) + d \cdot \expect [f((\cup_{i = 1}^\ell B_i) \mid O)\bigg]\\
	\geq{} &
	(d'\ell - d) \cdot f(O)
	\enspace,
\end{align*}
where the second inequality follows from the submodularity of $f$, and the last inequality follows from $f$'s non-negativity. Recall that there must be some strictly down-montone set $O$ that is also an optimal solution. Thus, when $\ell \geq 2d$, the approximation ratio of Algorithm~\ref{alg:repeat} is at most
\begin{align*}
	\frac{\ell \cdot (k + d + d'(\ell - 1) + 1 + 2\eps)}{d'\ell - d}
	\leq{} &
	\ell + \frac{\ell \cdot (k + 2d + 1 + 2\eps)}{d'\ell - d}\\
	={} &
	\ell + \frac{k + 2d + 1 + 2\eps}{\frac{1 - 1/d}{2 \ln 2} + \frac{d + 1}{2d} - d/\ell}
	=
	\ell + \frac{k + 2d + 1 + 2\eps}{\frac{1 + \ln 2}{2 \ln 2} - d^{-1}\cdot \frac{1 - \ln 2}{2 \ln 2} - d/\ell}\\
	\leq{} &
	\ell + \frac{(k + 2d + 1 + 2\eps) \cdot \big(1 + 2d^{-1}\cdot \frac{1 - \ln 2}{1 + \ln 2} + \frac{2d}{\ell} \cdot \frac{2 \ln 2}{1 + \ln 2}\big)}{\frac{1 + \ln 2}{2 \ln 2}}
	\enspace,
\end{align*}
where the second inequality holds because for every two positive values $a, b$ obeying $2b \leq a$, it holds that $\frac{1}{a - b} \leq \frac{1 + 2b/a}{a}$.
Choosing now $\ell = \lceil 4k^{2/3} \rceil$ and $d = 2k^{1/3}$, and recalling that $\eps$ is assigned a constant value from $(0, 1)$, we get that the approximation ratio of Algorithm~\ref{alg:repeat} is at most $\frac{2k \ln 2}{1 + \ln 2} + O(k^{2/3})$.

It remains to bound the time complexity of Algorithm~\ref{alg:repeat}. The Double Greedy algorithm of~\cite{buchbinder2015tight} runs in $O(|E|)$ time, and in Appendix~\ref{app:time_complexity}, we show that it is possible to simulate Algorithm~\ref{alg:1} in $O(\eps^{-1}|E|^4)$ time. Since Algorithm~\ref{alg:repeat} invokes each one of these algorithms $\ell = \lceil 3k^{2/3} \rceil$ times, we get that Algorithm~\ref{alg:repeat} can be implemented to run in
\[
	O(k^{2/3} \eps^{-1} |E|^4) = O(k^{2/3} |E|^4) = \Poly(|E|, k)
\]
time. When $k \leq 2|E|$, this implies the theorem. When $k > 2|E|$, the theorem still holds since in this case the na\"{i}ve algorithm that simply returns the best feasible singleton from $E$ already obeys all the properties guaranteed by the theorem.
\end{proof}


\appendix
\section{Polynomial Time Implementation of Algorithm~\ref{alg:1}} \label{app:time_complexity}

In Section~\ref{sec:local-search-algor}, we proved only that Algorithm~\ref{alg:1} terminated. Here, we show that it is possible to efficiently implement Algorithm~\ref{alg:1} in polynomial time. 
Recall that $L$ is the number of iterations that Algorithm~\ref{alg:1} performs. In Section~\ref{sec:local-search-algor}, we have only shown (in Observation~\ref{obs:final-set-L}) that $L$ is finite. One can note, however, that the number of iterations $i$ constructing a non-empty set $A_i$ is limited to $|E|$ because the sets $A_i$ are disjoint. Moreover, immediately after each iteration that constructs a non-empty set $A_{i'}$, and before starting any further iteration, we can already calculate the index $i$ of the next iteration in which $A_{i}$ will be non-empty. Thus, it is possible to accelerate the algorithm by skipping directly to iteration $i$, rather then going through all the iterations $i'+1, i'+2, \dotsc, i - 1$, and generating an empty set in each one of them. Formally, the index $i$ to which we should skip after iteration $i'$ is given by 
\[
	i = \lceil \log_2 (W\tau) - \log_2 w \rceil
	\enspace,
\]
where $w$ is the maximum marginal contribution $f(e \mid A_{\leq i'})$ of any element $e$ with $A_{\leq i'} + e \in \cI$. This formula is well-defined only for $w > 0$, but this is not a problem since the termination condition of Algorithm~\ref{alg:1} is equivalent to terminating when $w \leq 0$. The accelerated version of Algorithm~\ref{alg:1} obtained by implementing this observation is shown in Algorithm~\ref{alg:efficient}.
\begin{algorithm}
  Let $W \gets \max_{e \in E} f(\{e\} \mid \emptyset)$.\;
  Let $\alpha$ be a uniformly random value from $(0,1]$, and let $\tau \gets 2^{\alpha}$.\;
  Define $m_i \triangleq W \tau 2^{-i}$ for every integer $i \geq 0$.\;
  \BlankLine
  Let $i \gets 0$ and $A_{\leq 0} \gets \emptyset$.\;
  \While{there exists $e \in E \setminus A_{\leq i}$ such that $f(e \mid A_{\leq i}) > 0$ and $A\leq i + e \in \cI$}
  {
    Let $i' \gets i$\Comment*{Store the value of $i$ from the previous iteration.}
    Let $E_{i'} \gets \{e \in E \setminus A_{\leq i'} \mid A_{\leq i'} + e \in \cI\}$.\;
    Let $w_{i'} \gets \max_{e \in E_{i'}} f(e \mid A_{\leq i'})$.\;
    Update $i \gets \lceil \log_2 (W\tau) - \log_2 w_{i'} \rceil$.\label{line:next_iteration}\;
    Let $A_i \gets \emptyset$.\;
    \While{there is some $(m_{i}, \eps)$-improvement $(S,N)$ for $A_{\leq i'} \cup A_i$ with $N \subseteq A_i$\label{line:eff-while}}
    {
      $A_i \gets (A_i \setminus N) \cup S$.\label{line:eff-whileend}\;
    }
    Let $A_{\leq i} \gets A_{\leq i'} \cup A_i$.\;
  }
  \Return{$A_{\leq i}$}

\caption{\textsc{Efficient Greedy/Local-Search Hybrid Algorithm} $(E, f, \cI,\eps)$}
\label{alg:efficient}
\end{algorithm}

Similarly to Algorithm~\ref{alg:1}, the condition on Line~\ref{line:eff-while} of  Algorithm~\ref{alg:efficient} implies that the set $A_{\leq i} = A_{\leq i'} \cup A_i$ finalized at the end of each iteration carried out by the algorithm must be locally optimal with respect to $(m_i, \eps)$-improvements. In particular, the following direct analogue of Observation~\ref{obs:local-optimaly} holds.
\begin{observation}
\label{obs:eff-local-optimality}
For every iteration $i$ in which Algorithm~\ref{alg:efficient} produces a set $A_i$, we have $f(e \mid A_{\leq i}) < m_i$ for all $e \in E \setminus A_{\leq i}$ that obey $A_{\leq i} + e \in \cI$ at the end of iteration $i$.
\end{observation}

Our next objective is to show that Algorithm~\ref{alg:efficient} correctly simulates Algorithm~\ref{alg:1}. We assume in the proof that the two algorithms use the same method for finding the next $(m_i, \eps)$-improvement in each iteration of the inner loop, and we fix the values of any random bits used by this method (if it is randomized), as well as the value of $\alpha$. Once these values are fixed, showing that Algorithm~\ref{alg:efficient} correctly simulates Algorithm~\ref{alg:1} boils down to proving that they both return the same output set. Intuitively, to prove this, we need to show that the ``fast forward'' operation done by Line~\ref{line:next_iteration} of Algorithm~\ref{alg:efficient} works as expected by the intuition given at the beginning of this appendix. The next lemma is the first step towards this goal. It proves the simpler claim that the fast forward operation always increases $i$. Notice that this lemma is necessary for guaranteeing that each iteration of Algorithm~\ref{alg:efficient} has a unique index $i$.
\begin{lemma} \label{lem:i_increase}
Every execution of Line~\ref{line:next_iteration} of Algorithm~\ref{alg:efficient} strictly increases the value of $i$. Moreover, after $i$ is increased, we have $m_i \leq w_{i'} < m_{i-1} = 2m_i$
\end{lemma}
\begin{proof}
In the first execution of Line~\ref{line:next_iteration}, $w_{0} = W$, which implies that $i$ is increased from $0$ to
\[
	\lceil \log_2 (W\tau) - \log_2 W \rceil
	=
	\lceil \log_2 \tau \rceil
	=
	\lceil \alpha \rceil
	=
	1
	\enspace.
\]

Consider now an arbitrary execution of Line~\ref{line:next_iteration} after the first. At the start of the iteration of this execution, Algorithm~\ref{alg:efficient} considers the set $E_{i'} = \{e \in E \setminus A_{\leq i'} \mid A_{\leq i'} + e \in \cI\}$, where $A_{\leq {i'}}$ is the set produced by the previous iteration. By Observation~\ref{obs:eff-local-optimality}, we must have $f(e \mid A_{\leq i'}) < m_{i'}$ for all $e \in E_{i'}$, and hence, $w_{i'}< m_{i'}$. Thus, 
\[
	\log_2 (W\tau) - \log_2 w_{i'}
	>
	\log_2 (W\tau) - \log_2 m_{i'}
	=
	\log_2 (W\tau) - \log_2 (W\tau2^{-i'})
	=
	i'
	\enspace.
\]
The new value assigned to $i$ by Line~\ref{line:next_iteration} is the ceiling of the term on the left, which must be at least $i' + 1$ since $i'$ is an integer.

For the second part of the claim, note that by definition of $m_i$, 
\[
	m_i
	=
	W\tau 2^{-i}
	=
	W\tau 2^{-\lceil \log_2 (W\tau) - \log_2 w_{i'} \rceil}
	\leq
	W\tau 2^{\log_2 w_{i'} - \log_2 (W\tau)}
	=
	w_{i'}
	\enspace.
\]
Similarly,
\[
	m_{i-1}
	=
	W\tau 2^{-(i-1)}
	=
	W\tau 2^{-(\lceil \log_2 (W\tau) - \log_2 w_{i'} \rceil - 1)}
	>
	W\tau 2^{\log_2 w_{i'} - \log_2 (W\tau)}
	=
	w_{i'}\enspace,
\]
where the strict inequality holds because Algorithm~\ref{alg:efficient} terminates without making a single iteration when $W \leq 0$.
\end{proof}

%

In Section~\ref{sec:local-search-algor}, we argue that each iteration of Algorithm~\ref{alg:1} must terminate. This argument applies also to Algorithm~\ref{alg:efficient}. Let $C$ be the set of all values from $[L]$ that are assigned to $i$ by Algorithm~\ref{alg:efficient} at some point. Clearly, Algorithm~\ref{alg:efficient} constructs sets $A_j$ and $A_{\leq j}$ for each $j \in C$, while Algorithm~\ref{alg:1} constructs sets $A_j$ and $A_{\leq j}$ for every $j \in [L]$. The following lemma gives additional properties of the set $C$.

\begin{lemma}\label{lem:simulation}
For all $j \in [L]$, $j \in C$ if and only if Algorithm~\ref{alg:1} produces a set $A_{j} \neq \emptyset$. Moreover, for all $j \in C \cup \{0\}$,  Algorithms~\ref{alg:repeat} and~\ref{alg:efficient} construct the same set $A_{\leq j}$.
\end{lemma}
\begin{proof}
We proceed by induction on $j$. For $j = 0$, the lemma simply states that $A_{\leq 0}$ is identical under both Algorithm~\ref{alg:1} and Algorithm~\ref{alg:efficient}, which is the case since both algorithms set it to be the empty set.

Suppose $j \geq 1$, and let $j' < j$ be the next largest index in $C \cup \{0\}$. By the induction hypothesis, both algorithms agree on the value of $A_{\leq j'}$, and for any iteration $i$ with $j' < i < j$, Algorithm~\ref{alg:1} must have $A_i = \emptyset$. Thus, $A_{\leq j-1} = A_{\leq j'} \cup (\cup_{i = j'+1}^{j-1}A_i) = A_{\leq j'}$. We now need to consider two cases, depending on whether $j \in C$ or not. Assume first that $j \in C$. Then, both algorithms produce the same set $A_{j}$ because (i) $A_{\leq j - 1} = A_{\leq j'}$, (ii) we fixed the value of $\alpha$, and (iii) we assume that both algorithms use the same method for choosing $(m_i, \eps)$-improvements and we fixed the random bits used by this method (if any). This implies that both algorithms set $A_{\leq j}$ to the same value because Algorithm~\ref{alg:1} sets $A_{\leq j}$ to $A_j \cup A_{\leq j - 1}$ and Algorithm~\ref{alg:efficient} sets $A_{\leq j}$ to $A_j \cup A_{\leq j'}$. We now show that $A_j \neq \emptyset$. Since $j \in C$, Algorithm~\ref{alg:efficient} has $w_{j'} \geq m_{j}$ by Lemma~\ref{lem:i_increase}. Thus, there must be some $x \in E \setminus A_{\leq j'}$ with $A_{\leq j'} + x \in \cI$ and $f(x \mid A_{\leq j'}) \geq m_i$. If $A_j = \emptyset$, then $A_{\leq j} = A_{\leq j'} \cup \emptyset = A_{\leq j'}$, and hence, it also holds that $A_{\leq j} + x \in \cI$ and $f(x \mid A_{\leq j}) \geq m_j$, contradicting local optimality at the end of the iteration (Observation~\ref{obs:eff-local-optimality}). 

Consider now the case that $j \not \in C$, and assume towards a contradiction that Algorithm~\ref{alg:1} produces a set $A_{j} \neq \emptyset$. This means that Algorithm~\ref{alg:1} performed at least one $(m_j, \eps)$-improvement in iteration $j$, and thus, there exists an element $x \in E \setminus A_{\leq j - 1} = E \setminus A_{\leq j'}$ such that $A_{\leq j'} + x = A_{\leq j - 1} + x \in \cI$ and $f(x \mid A_{\leq j'}) = f(x \mid A_{\leq j - 1}) \geq m_j$. Let us denote now by $\hat{\jmath}$ the value that $i$ gets in Algorithm~\ref{alg:efficient} immediately after the value $j'$. The existence of $x$ guarantees that such a value $\hat{\jmath}$ exists, and by the definition of $j'$, $\hat{\jmath}$ must be at least $j + 1$. Thus, by Lemma~\ref{lem:i_increase},
\[
	w_{j'}
	<
	m_{\hat{j} - 1}
	\leq
	m_j
	\leq
	f(x \mid A_{\leq j'})
	\enspace,
\]
which is a contradiction since the definition $w_{j'}$ guarantees that $w_{j'} \geq f(x \mid A_{\leq j'})$.
\end{proof}

\begin{lemma}
\label{lem:final-set}
Algorithm~\ref{alg:efficient} returns the same set $A_{\leq L}$ as Algorithm~\ref{alg:1}, and the final value of $i$ in this algorithm is $L$.
\end{lemma}
\begin{proof}
We begin by proving that both algorithms produce the same set $A_{\leq L}$. If $L = 0$, then by definition, both algorithms set $A_{\leq L} = A_{\leq 0} = \emptyset$. Otherwise, if $L \geq 1$, then by Observation~\ref{obs:final-set-L}, Algorithm~\ref{alg:1} constructs a non-empty set $A_L$, and thus, by Lemma~\ref{lem:simulation}, Algorithm~\ref{alg:efficient} constructs the same set $A_{\leq L}$ as Algorithm~\ref{alg:1}.

Let us now explain why the lemma follows from the fact that both algorithms produce the same set $A_{\leq L}$.
Observation~\ref{obs:final-set-L} shows that $f(e \mid A_{\leq L}) \leq 0$ for all $e \in E \setminus A_{\leq L}$ with $A_{\leq L} + e \in \cI$. Thus, the main loop of Algorithm~\ref{alg:efficient} must terminate immediately after the iteration constructing $A_{\leq L}$ (or before the first iteration if $L = 0$), which makes this algorithm return $A_{\leq L}$.
\end{proof}

We now get to our final result regarding Algorithm~\ref{alg:efficient}

\begin{proposition}
Algorithm~\ref{alg:efficient} has the same output distribution as Algorithm~\ref{alg:1}, and it runs in $O(\eps^{-1}|E|^4)$ time.
\end{proposition}
\begin{proof}
The first part of the proposition holds since Lemma~\ref{lem:final-set} guarantees that the two algorithms produce the same output set when they draw the same random bits. Thus, in the rest of this proof, we concentrate on proving that Algorithm~\ref{alg:1} can be implemented to run in $O(\eps^{-1}|E|^4)$ time.

By Lemma~\ref{lem:final-set}, $i$ never reaches any value larger than $L$. Together with the definition of $C$, this implies that the number of iterations Algorithm~\ref{alg:efficient} performs is $|C|$. By Lemma~\ref{lem:simulation}, each value $i \in C$ corresponds to a non-empty set $A_i$ in Algorithm~\ref{alg:1}, and since these sets are disjoint by construction, $|C| \leq |E|$. Thus, Algorithm~\ref{alg:efficient} makes at most $|E|$ iterations. Each such iteration takes at most $O(|E|)$ time (the time required to compute $w_i$ and update $i$) plus the time required to search and apply the $(m_i, \eps)$ improvements. Thus, the running time of Algorithm~\ref{alg:efficient} is $O(|E|^2)$ plus the total time spent searching for and applying these improvements across all iterations. 

Let us denote by $M$ the number of improvements found throughout the algorithm. Notice that in each iteration of Algorithm~\ref{alg:efficient}, the number of times that the algorithms looks for an improvement is larger by exactly $1$ compared to the number of improvements found in this iteration. Thus, Algorithm~\ref{alg:efficient} looks for improvements at most $M + |E|$ times. Looking for an improvement can be done in $O(|E|^3)$ time by simply enumerating over all the possible options for a set $S \subseteq E \setminus (A_i \cup A_{\leq i'})$ of size at most $2$ and a set $N \subseteq A_i$ of size at most $1$. Once an improvement is found, applying it (i.e., replacing $A_i$ with $(A_i \setminus N) \cup S$) can be done in $O(1)$ time, and thus, by the above discussion, the time complexity of Algorithm~\ref{alg:efficient} is
\begin{multline} \label{eq:time_complexity}
	O(|E|^2 + \text{\{time to look for an improvement\}} \cdot (M + |E|) + \text{\{time to apply an improvment\}} \cdot M)\\
	=
	O(|E|^2 + |E|^3(M + |E|) + M) = O(|E|^4 + |E|^3M)
	\enspace.
\end{multline}

It remains to bound $M$. Consider the results of applying each of the three types of improvements in Definition~\ref{def:improvement} in some iteration $i$.
\begin{itemize}
\item Applying the first type of improvement increases $|A_{i}|$ by 1 and increases $f(A_{\leq i'}\cup A_i)$ by at least $\theta > 0$.
\item Applying the second type of improvement does not change $|A_i|$ and increases $f(A_{\leq i'})$ by at least $\eps \theta$.
\item Applying the third type of improvement increases $|A_i|$ by 1 and does not decrease the value of $A_{\leq i'} \cup A_i$. This latter claim follows  from the submodularity of $f$ and Lemma~\ref{lem:i_increase}, which together give
  \begin{multline*}
    f((A_{\leq i'} \cup A_i \cup S) \setminus N) - f(A_{\leq i'} \cup A_i) 
    \\
    \begin{aligned}
      &=
      f(x_1 \mid A_{\leq i'} \cup A_i - y) 
      + f(x_2 \mid A_{\leq i'} \cup A_i - y + x_1) 
      - f(y \mid A_{\leq i'} \cup A_i - y) 
      \\
      &\geq 
      f(x_1 \mid A_{\leq i'} \cup A_i) 
      + f(x_2 \mid A_{\leq i'} \cup A_i + x_1) 
      - f(y \mid A_{\leq i'})
      \\
      &\geq
      2m_i - w_{i'} > 0
    \end{aligned}
  \end{multline*}
\end{itemize}
(recall that $x_1$ and $x_2$ are labels given by Definition~\ref{def:improvement} for the two elements of $S$, and $y$ is the single element of $N$). Thus, if we consider $|A_i|$ and $f(A_{\leq i'} \cup A_i)$ as two potentials then no $(m_i,\eps)$-improvement ever decreases these potentials, and each such improvement either increases $|A_i|$ by $1$ or increases $f(A_{\leq i'} \cup A_i)$ by at least $\eps m_{i}$. Let us denote by $\hat{A}_i$ the final value of the set $A_i$ in iteration $i$. Then, the potential $|A_i|$ is non-negative and upper bounded by $|\hat{A}_i|$, and therefore at most $|\hat{A}_i|$ $(m_i, \eps)$-improvements can increase it by $1$. Similarly, by the submodularity of $f$, the potential $f(A_{\leq i'} \cup A_i)$ starts with a value of $f(A_{\leq i'})$ and ends up with the value
\[
	f(A_{\leq i'} \cup \hat{A}_i)
	\leq
	f(A_{\leq i'}) + \sum_{e \in \hat{A}_i} f(e \mid A_{\leq i'})
	\leq
	f(A_{\leq i'}) + |\hat{A}_i| \cdot w_{i'}
	\leq
	f(A_{\leq i'}) + 2|\hat{A}_i| \cdot m_i
	\enspace,
\]
where the last inequality employs Lemma~\ref{lem:i_increase}.
Therefore, at most $2\eps^{-1}|\hat{A}_i|$ improvements can increase this potential by $\eps m_{i}$. Thus, the total number of $(m_i, \eps)$-improvements of both kinds in iteration $i$ of Algorithm~\ref{alg:efficient} is upper bounded by $(1 + 2\eps^{-1})|\hat{A}_i|$. Since the sets $A_i$ are kept disjoint by Algorithm~\ref{alg:efficient}, summing the last expression over all iterations of the algorithm yields $M \leq (1 + 2\eps^{-1})|E| = O(\eps^{-1}|E|)$. The proposition now follows by plugging this bound on $M$ into \eqref{eq:time_complexity}.
\end{proof}


\section{Analyzing the Distribution of \texorpdfstring{$r_o$}{r\_o}} \label{apx:r_is_shift}

The following lemma analyzes the distribution of $r_o$ for an arbitrary element $o \in O$. Recall that $O$ is an arbitrary strictly down-monotone feasible set.

\begin{lemma} \label{lem:shift}
For each $o \in O$, the value $r_o$ defined in Section~\ref{sec:main-charging-scheme} is equal to $2^\beta$ for a random variable $\beta$ distributed uniformly at random in the range $[0,1)$.
\end{lemma}
\begin{proof}
Recall that for all $i \geq 0$, $m_i = W \tau 2^{-i} = W 2^{\alpha - i}$, where $\alpha$ is chosen uniformly at random from $(0,1]$, and note that $r_o = \min\{m_i/u(o) \mid m_i \geq u(o)\}$. Let 
\begin{align*}
i^* &= \lfloor \log_2W - \log_2u(o) \rfloor + 1 \enspace, \quad \text{and} \\
\alpha^* &= i^* - (\log_2 W - \log_2 u(o))\enspace.
\end{align*}
Then, we have $\alpha^* \in (0,1]$, and $i^*$ is a positive integer because the submodularity of $f$ guarantees that $u(o) \leq f(o \mid \emptyset) \leq W$. Additionally, since Observation~\ref{obs:positive_weights_diff} guarantees that $u(o) > 0$, we get that $W$ must be positive as well. Therefore, it holds that $W2^{\alpha^* - i^*} = u(o)$, $W2^{-i^*} < u(o)$, and $W2^{-i^*+1} \geq u(o)$.

The above inequalities imply that, for all $\alpha \in (0, \alpha^*)$,
\begin{align*}
m_{i^*} &= W 2^{\alpha-i^*} < W 2^{\alpha^* - i^*} = u(o) \enspace, \quad \text{and}\\
m_{i^* - 1} &= W 2^{\alpha-i^* + 1} > W 2^{-i^* + 1}  \geq u(o) \enspace,
\end{align*}
and hence, $r_o = m_{i^* - 1}/u(o) = (W 2^{\alpha -i^* + 1})/(W2^{\alpha^*-i^*}) = 2^{\alpha- \alpha^*+1}$. Similarly, for all $\alpha \in [\alpha^*, 1]$, 
\begin{align*}
m_{i^*} &= W 2^{\alpha - i^*} \geq W 2^{\alpha^* - i^*} = u(o) \enspace, \quad \text{and} \\
m_{i^* + 1} &= W 2^{\alpha-i^* - 1} \leq W 2^{-i^*} < u(o) \enspace,
\end{align*}
and hence, $r_o = m_{i^*}/u(o) = (W 2^{\alpha-i^*})/(W2^{\alpha^* - i^*}) = 2^{\alpha - \alpha^*}$. Altogether, we then have $r_o = 2^\beta$, where
\begin{equation*}
\beta = 
\begin{cases}\alpha - \alpha^* + 1, &\alpha \in (0,\alpha^*) \\
\alpha - \alpha^*, &\alpha \in [\alpha^*, 1]\,.
\end{cases}
\end{equation*}

To complete the proof of the lemma, it only remains to show that $\beta$ is distributed uniformly at random in the range $[0,1)$. To show this, we note that for any value $\gamma \in [0,1-\alpha^*]$,
\begin{equation*}
\Pr[\beta \leq \gamma] = \Pr[\alpha^* \leq \alpha \leq \alpha^*+\gamma] = (\alpha^* + \gamma) - \alpha^* = \gamma\enspace,
\end{equation*}
and for any value $\gamma \in (1-\alpha^*, 1)$,
\begin{align*}
\Pr[\beta \leq \gamma] ={} & \Pr[\beta \leq 1-\alpha^*] + \Pr[1-\alpha^* < \beta \leq \gamma]\\
={} & 1 - \alpha^* + \Pr[0 < \alpha \leq \gamma - 1 + \alpha^*]
= 1 - \alpha^* + (\gamma - 1 + \alpha^*) = \gamma\enspace.
\qedhere
\end{align*}
\end{proof}


\bibliographystyle{plain}
\bibliography{MatroidIntersection}

\begin{thebibliography}{10}

\bibitem{arkin1998local}
Esther~M. Arkin and Refael Hassin.
\newblock On local search for weighted $k$-set packing.
\newblock {\em Mathematics of Operations Research}, 23(3):640--648, 1998.

\bibitem{barvinok1995new}
Alexander~I. Barvinok.
\newblock New algorithms for linear $k$-matroid intersection and matroid
  $k$-parity problems.
\newblock {\em Mathematical Programming}, 69:449--470, 1995.

\bibitem{berman2000approximation}
Piotr Berman.
\newblock A $d/2$ approximation for maximum weight independent set in $d$-claw
  free graphs.
\newblock {\em Nordic J. Computing}, 7(3):178--184, 2000.

\bibitem{berman2003optimizing}
Piotr Berman and Piotr Krysta.
\newblock Optimizing misdirection.
\newblock In {\em {ACM-SIAM} Symposium on Discrete Algorithms ({SODA})}, pages
  192--201, 2003.

\bibitem{buchbinder2015tight}
Niv Buchbinder, Moran Feldman, Joseph Naor, and Roy Schwartz.
\newblock A tight linear time (1/2)-approximation for unconstrained submodular
  maximization.
\newblock {\em {SIAM} J. Computing}, 44(5):1384--1402, 2015.

\bibitem{chakrabarti2015submodular}
Amit Chakrabarti and Sagar Kale.
\newblock Submodular maximization meets streaming: matchings, matroids, and
  more.
\newblock {\em Mathematical Programming}, 154(1-2):225--247, 2015.

\bibitem{chandra2001greedy}
Barun Chandra and Magn{\'{u}}s~M. Halld{\'{o}}rsson.
\newblock Greedy local improvement and weighted set packing approximation.
\newblock {\em J. Algorithms}, 39(2):223--240, 2001.

\bibitem{chekuri2015streaming}
Chandra Chekuri, Shalmoli Gupta, and Kent Quanrud.
\newblock Streaming algorithms for submodular function maximization.
\newblock In {\em International Colloquium on Automata, Languages, and
  Programming {ICALP}}, pages 318--330, 2015.

\bibitem{calinescu2011maximizing}
Gruia C{u{a}}linescu, Chandra Chekuri, Martin P{\'{a}}l, and Jan Vondr{\'{a}}k.
\newblock Maximizing a monotone submodular function subject to a matroid
  constraint.
\newblock {\em {SIAM} J. Computing}, 40(6):1740--1766, 2011.

\bibitem{cunningham1986improved}
William~H. Cunningham.
\newblock Improved bounds for matroid partition and intersection algorithms.
\newblock {\em {SIAM} J. Computing}, 15(4):948--957, 1986.

\bibitem{cygan2013improved}
Marek Cygan.
\newblock Improved approximation for 3-dimensional matching via bounded
  pathwidth local search.
\newblock In {\em {IEEE} Symposium on Foundations of Computer Science, {FOCS}},
  pages 509--518, 2013.

\bibitem{edmonds2003submodular}
Jack Edmonds.
\newblock Submodular functions, matroids, and certain polyhedra.
\newblock In Michael J{\"u}nger, Gerhard Reinelt, and Giovanni Rinaldi,
  editors, {\em Combinatorial Optimization --- Eureka, You Shrink!: Papers
  Dedicated to Jack Edmonds}, pages 11--26. Springer Berlin Heidelberg, Berlin,
  Heidelberg, 2003.

\bibitem{feldman2023how}
Moran Feldman, Christopher Harshaw, and Amin Karbasi.
\newblock How do you want your greedy: Simultaneous or repeated?
\newblock {\em J. Machine Learning Research}, 24:72:1--72:87, 2023.

\bibitem{feldman2011improved}
Moran Feldman, Joseph Naor, Roy Schwartz, and Justin Ward.
\newblock Improved approximations for $k$-exchange systems - (extended
  abstract).
\newblock In {\em European Symposium on Algorithms ({ESA})}, pages 784--798,
  2011.

\bibitem{fisher1978analysis}
Marshall~L. Fisher, George~L. Nemhauser, and Laurence~A. Wolsey.
\newblock An analysis of approximations for maximizing submodular set
  functions---{II}.
\newblock {\em Mathematical Programming Studies}, 8:73--87, 1978.

\bibitem{furer2014approximating}
Martin F{\"{u}}rer and Huiwen Yu.
\newblock Approximating the $k$-set packing problem by local improvements.
\newblock In {\em International Symposium on Combinatorial Optimization
  {ISCO}}, pages 408--420, 2014.

\bibitem{greene1975some}
Curtis Greene and Thomas~L. Magnanti.
\newblock Some abstract pivot algorithms.
\newblock {\em {SIAM} J. Applied Mathematics}, 29(3):530--539, 1975.

\bibitem{gupta2010constrained}
Anupam Gupta, Aaron Roth, Grant Schoenebeck, and Kunal Talwar.
\newblock Constrained non-monotone submodular maximization: Offline and
  secretary algorithms.
\newblock In {\em International Workshop on Internet and Network Economics
  ({WINE})}, pages 246--257, 2010.

\bibitem{hazan2006complexity}
Elad Hazan, Shmuel Safra, and Oded Schwartz.
\newblock On the complexity of approximating $k$-set packing.
\newblock {\em Computational Complexity}, 15(1):20--39, 2006.

\bibitem{huang2023matroid}
Chien{-}Chung Huang and Fran{\c{c}}ois Sellier.
\newblock Matroid-constrained vertex cover.
\newblock {\em Theoretical Computer Science}, 965:113977, 2023.

\bibitem{iwata2022weighted}
Satoru Iwata and Yusuke Kobayashi.
\newblock A weighted linear matroid parity algorithm.
\newblock {\em {SIAM} J. Computing}, 51(2):17--238, 2022.

\bibitem{jenkyns1976efficacy}
T.A. Jenkyns.
\newblock The efficacy of the "greedy" algorithm.
\newblock In {\em Southeastern Conf. on Combinatorics, Graph Theory and
  Computing}, pages 341--350, 1976.

\bibitem{jensen1982complexity}
Per~M. Jensen and Bernhard Korte.
\newblock Complexity of matroid property algorithms.
\newblock {\em {SIAM} J. Computing}, 11(1):184--190, 1982.

\bibitem{karp72reducibility}
Richard~M. Karp.
\newblock Reducibility among combinatorial problems.
\newblock In {\em Symposium on the Complexity of Computer Computations}, pages
  85--103, 1972.

\bibitem{korte1978greedy}
Bernhard Korte and Dirk Hausmann.
\newblock An analysis of the greedy heuristic for independence systems.
\newblock {\em Annals of Discrete Mathematics}, 2:65--74, 1978.

\bibitem{lawler1976combinatorial}
E.~Lawler.
\newblock {\em Combinatorial optimization - networks and matroids}.
\newblock Holt, Rinehart and Winston, New York, 1976.

\bibitem{lee2025asymptotically}
Euiwoong Lee, Ola Svensson, and Theophile Thiery.
\newblock Asymptotically optimal hardness for $k$-set packing and k-matroid
  intersection.
\newblock In {\em {ACM} Symposium on Theory of Computing ({STOC})}, pages
  54--61, 2025.

\bibitem{lee10submodular}
Jon Lee, Maxim Sviridenko, and Jan Vondr{\'{a}}k.
\newblock Submodular maximization over multiple matroids via generalized
  exchange properties.
\newblock {\em Mathematics of Operations Research}, 35(4):795--806, 2010.

\bibitem{lee2013matroid}
Jon Lee, Maxim Sviridenko, and Jan Vondr{\'{a}}k.
\newblock Matroid matching: The power of local search.
\newblock {\em {SIAM} J. Computing}, 42(1):357--379, 2013.

\bibitem{lovasz1980matroid}
L{\'{a}}szl{\'{o}} Lov{\'{a}}sz.
\newblock Matroid matching and some applications.
\newblock {\em J. of Combinatorial Theory, Series B}, 28(2):208--236, 1980.

\bibitem{lovasz1981matroid}
L{\'{a}}szl{\'{o}} Lov{\'{a}}sz.
\newblock The matroid matching problem.
\newblock In L.~Lov{\'{a}}sz and V.~T. S{\'{o}}s, editors, {\em Algebraic
  Methods in Graph Theory, Vol. II}, pages 495--517. North-Holland, Amsterdam,
  1981.

\bibitem{marx2009parameterized}
D{\'{a}}niel Marx.
\newblock A parameterized view on matroid optimization problems.
\newblock {\em Theoretical Computer Science}, 410(44):4471--4479, 2009.

\bibitem{mirzasoleiman2016fast}
Baharan Mirzasoleiman, Ashwinkumar Badanidiyuru, and Amin Karbasi.
\newblock Fast constrained submodular maximization: Personalized data
  summarization.
\newblock In {\em International Conference on Machine Learning ({ICML})}, pages
  1358--1367, 2016.

\bibitem{neuwohner2021improved}
Meike Neuwohner.
\newblock An improved approximation algorithm for the maximum weight
  independent set problem in $d$-claw free graphs.
\newblock In {\em International Symposium on Theoretical Aspects of Computer
  Science, {STACS}}, pages 53:1--53:20, 2021.

\bibitem{neuwohner2023passing}
Meike Neuwohner.
\newblock Passing the limits of pure local search for weighted $k$-set packing.
\newblock In {\em {ACM-SIAM} Symposium on Discrete Algorithms ({SODA})}, pages
  1090--1137, 2023.

\bibitem{neuwohner2024limits}
Meike Neuwohner.
\newblock The limits of local search for weighted $k$-set packing.
\newblock {\em Mathematical Programming}, 206(1):389--427, 2024.

\bibitem{schrijver2003combinatorial}
Alexander Schrijver.
\newblock {\em Combinatorial Optimization: Polyhedra and Efficiency}, volume~24
  of {\em Algorithms and Combinatorics}.
\newblock Springer-Verlag, Berlin Heidelberg, 2003.

\bibitem{singer2025better}
Neta Singer and Theophile Thiery.
\newblock Better approximation for weighted k-matroid intersection.
\newblock In {\em {ACM} Symposium on Theory of Computing ({STOC})}, pages
  1142--1153, 2025.

\bibitem{sviridenko2013large}
Maxim Sviridenko and Justin Ward.
\newblock Large neighborhood local search for the maximum set packing problem.
\newblock In {\em International Colloquium on Automata, Languages, and
  Programming ({ICALP})}, pages 792--803, 2013.

\bibitem{thiery2023thesis}
Theophile Thiery.
\newblock {\em Approximation Algorithms for Independence Systems}.
\newblock {Ph.D.} thesis, Queen Mary University of London, London, United
  Kingdom, 2023.

\bibitem{thiery2023improved}
Theophile Thiery and Justin Ward.
\newblock An improved approximation for maximum weighted $k$-set packing.
\newblock In {\em {ACM-SIAM} Symposium on Discrete Algorithms, {SODA}}, pages
  1138--1162, 2023.

\bibitem{ward2012approximation}
Justin Ward.
\newblock A $(k+3)/2$-approximation algorithm for monotone submodular k-set
  packing and general k-exchange systems.
\newblock In {\em International Symposium on Theoretical Aspects of Computer
  Science ({STACS})}, pages 42--53, 2012.

\end{thebibliography}

\end{document}